\documentclass{article}

\usepackage{color,ifthen,url,amssymb,amsmath,amsthm}
\usepackage{fullpage}
\usepackage{graphicx}
\usepackage{subfigure}
\usepackage{times}
\usepackage{float}

\newcommand{\comment}[1]{}

\floatstyle{ruled}
\newfloat{algorithm}{t}{loa}
\floatname{algorithm}{Algorithm}

\newcommand{\hf}{\hat{f}}
\newcommand{\transpose}{^{\textsf{T}}}
\newcommand{\ind}{\chi}
\newcommand{\targetf}{{f^*}}
\newcommand{\poly}{\operatorname{poly}}
\newcommand{\bN}{\mathbb{N}}        
\newcommand{\bR}{\mathbb{R}}
\newcommand{\reals}{{\mathbb{R}}}

\newcommand{\cF}{\mathcal{F}} 
\newcommand{\cP}{\mathcal{P}}
\newcommand{\cL}{\mathcal{P}} 

\newcommand{\cS}{\mathcal{S}}
\newcommand{\cSz}{{\mathcal{S}}_{0}}
\newcommand{\cSnz}{{\mathcal{S}}_{\neq 0}}
\newcommand{\cY}{\mathcal{Y}}
\newcommand{\cZ}{\mathcal{Z}}
\newcommand{\cU}{\mathcal{U}}
\renewcommand{\ell}{m}

\newcommand{\demandSet}[1]{\mathcal{D}^{#1}}
\newcommand{\eps}{\varepsilon}
\newcommand{\prob}[1]{\operatorname{Pr}\left[\,#1\,\right]}               
\newcommand{\probg}[2]{\operatorname{Pr}\left[\,#1 \:\mid\: #2\,\right]}  
\newcommand{\probover}[2]{\operatorname{Pr}_{#1}\left[\,#2 \,\right]}     
\newcommand{\expect}[1]{\operatorname{\bf E}\left[\,#1\,\right]}          

\newcommand{\cA}{\mathcal{A}}

\newcommand{\trainS}{{\cal S}}
\newcommand{\cT}{\mathcal{T}}

\newcommand{\abs}[1]{\lvert #1 \rvert}
\newcommand{\OR}{{\textsc{SUM}\ }}
\newcommand{\XOR}{{\textsc{MAX}\ }}
\newcommand{\XOS}{\textsc{XOS}}
\newcommand{\XOSLong}{\XOS} 
\newcommand{\OXS}{\textsc{OXS}}
\newcommand{\OXSLong}{\OXS} 
\newcommand{\GS}{\textsc{GS}\ }

\newcommand{\sgn}{{\operatorname{sgn}}}
\newcommand{\card}[1]{\abs{#1}}

\newtheorem{Theorem}{Theorem}

\newtheorem{Definition}{Definition}
\newtheorem{Lemma}{Lemma}
\newtheorem{Claim}{Claim}
\newtheorem{claim}{Claim}

\newenvironment{proofsketch}{{\em Proof Sketch:}}{\qed}
\newcommand{\Union}{\cup}
\newcommand{\set}[1]{\left \{ #1 \right \}}                     

\newcommand{\targetfpower}[1]{(\targetf(#1))^{p}}

\newcommand{\setst}[2]{\left\{\; #1 \,:\, #2 \;\right\}}        

\newcommand{\setfamily}{family}

\title{Learning Valuation Functions}

\author{Maria Florina Balcan\thanks{
Georgia Institute of Technology,
\texttt{\small ninamf@cc.gatech.edu}}
\qquad
Florin Constantin\thanks{
Georgia Institute of Technology,
\texttt{\small florin@cc.gatech.edu}}
\qquad
Satoru Iwata\thanks{
Kyoto University,
\texttt{\small iwata@kurims.kyoto-u.ac.jp}}
\qquad
Lei Wang\thanks{
Georgia Institute of Technology,
\texttt{\small leiwang2007@gatech.edu}}
}

\date{}

\begin{document}
\maketitle
\thispagestyle{empty}
\setcounter{page}{0}

\begin{abstract}
A core element of microeconomics and game theory is that consumers
have valuation functions over bundles of goods and that these
valuation functions drive their purchases.  In particular, the value
assigned to a bundle need not be the sum of values on the individual
items but rather is often a more complex function of how the items relate.
The literature considers a hierarchy of valuation classes that includes
subadditive, \XOS\ (i.e. fractionally subadditive), submodular, and \OXS\ valuations.  Typically it
is assumed that these valuations are known to the center or that they come from a
known distribution.
Two recent lines of work, by Goemans et al. (SODA 2009) and by Balcan and Harvey (STOC 2011),  have considered a more realistic setting in which  valuations are learned from data,
focusing specifically on submodular functions.

In this paper we consider the approximate learnability of valuation functions at all levels in the hierarchy.
We first study their learnability in the distributional learning (PAC-style) setting due to Balcan and Harvey (STOC 2011).
We provide nearly tight lower and upper bounds of $\tilde{\Theta}(n^{1/2})$ on
 the approximation factor for learning \XOS\ and subadditive valuations,
both important classes that are strictly more general than submodular valuations.
Interestingly, we show that the $\tilde{\Theta}(n^{1/2})$ lower bound can be circumvented for \XOS\ functions of polynomial complexity; we provide an algorithm for learning
the class of \XOS\ valuations
 with a representation of polynomial size to within an $O(n^{\epsilon})$ approximation factor in running time $n^{O(1/\eps)}$ for any $\epsilon > 0$.
 We also establish learnability and hardness results for
 subclasses of the class of submodular valuations,
 i.e. gross substitutes valuations and interesting subclasses of \OXS\ valuations.

In proving our results for the distributional learning setting,
we provide novel  structural results for all these classes of valuations.
We show the implications of these results for the learning everywhere with value queries model, considered by Goemans et al. (SODA 2009).

Finally, we also introduce a more realistic variation of these models for economic settings,
in which information on the value of a bundle $S$ of goods can only be inferred
based on whether $S$ is purchased or not at a specific price.
 We provide lower and upper bounds for learning both in the distributional setting and with value queries.
\end{abstract}

\newcommand{\thispaper}{{[\small this paper]}}
\newcommand{\folklore}{{[\small folklore]}}

\newpage

\section{Introduction}

A central problem in commerce is understanding one's customers.  Whether
for assigning prices to goods, for deciding how to bundle products, or
for estimating how much inventory to carry, it is critical for a
company to understand its customers' preferences.  In Economics and
Game Theory, these preferences are typically modeled as valuations, or monotone set
functions, over subsets of goods.
It is usually assumed that consumers' valuations are known in
advance to the company, or that they are drawn from a known distribution.
In practice, however, these valuations must be learned.
For example, given past data of customer purchases of different bundles, a retailer would like to estimate how much a
(typical) customer would be willing to pay for new packages of goods
that become available.
Companies may also conduct surveys querying customers about their valuations\footnote{See e.g. \url{http://bit.ly/ls774D} for an example of an airline asking customers for a ``reasonable'' price for in-flight Internet.
}.

 Motivated by such scenarios, in this paper we investigate the learnability
 of classes of functions commonly used to model
consumers' valuations.
 In particular, we focus on a wide class of valuations
expressing ``no complementarities'':
 the value of the union of two disjoint bundles
 is no more than the sum of the values on each bundle ---
we henceforth use
 the standard optimization terminology {\em subadditive valuations}.
We provide upper and lower bounds on the learnability of valuation classes
in a popular hierarchy~\cite{GulS99WalrasianEwGS,LehmannLN02CombinatorialAwDMU,book07,Sandholm99AlgorithmfOWDiCA},
with submodular functions (the only class with similar extant results~\cite{BalcanH10LearningSF,nick09}) halfway in the hierarchy:
$$ \textit{OXS} \subsetneq \textit{gross substitutes}\ \subsetneq \text{submodular} \text{ \small (the only related learnability results~\cite{BalcanH10LearningSF,nick09}) } \subsetneq
\textit{XOS} \subsetneq \textit{subadditive}$$

We analyze the learnability of these classes in the
natural PMAC model~\cite{BalcanH10LearningSF} 
for approximate distributional learning. 
In this model, a learning algorithm is given a collection $\trainS = \{S_{1}, \dots, S_{\ell}\}$
of polynomially many labeled examples drawn i.i.d.~from
some fixed, but unknown, distribution $D$ over points (sets) in $2^{[n]}$.
The points are labeled by a fixed, but unknown, target function $\targetf : 2^{[n]} \rightarrow \bR_+$.
The goal is to output  in polynomial time, with high probability,  a hypothesis function $f$
that is a good multiplicative   approximation for $\targetf$  over most sets with respect to $D$. More formally, we want:
$$
\probover{S_1,\ldots,S_{\ell} \sim D}{~
    \probover{S \sim D}{
        f(S) \leq \targetf(S) \leq \alpha f(S)
    } ~\geq~ 1-\epsilon
~} ~\geq~ 1-\delta
$$ for an algorithm that uses $\ell= \poly(n, \frac{1}{\eps},\frac{1}{\delta})$ samples and that runs in $\poly(n, \frac{1}{\eps},\frac{1}{\delta})$ time.
In contrast, the classical PAC model~\cite{Valiant:acm84} requires
predicting \emph{exactly} (i.e. $\alpha\!=\!1$) with high probability the values of $\targetf$
 over most sets with respect to $D$. 
Thus the PMAC model can be viewed as an approximation-algorithms
extension of the traditional PAC model.

Our main results in the PMAC-learning model are for superclasses of submodular valuations, namely
subadditive valuations and \XOS~\cite{DobzinskiNS05ApproximationAfCAwCFB,DobzinskiNS06TruthfulRMfCA,LehmannLN02CombinatorialAwDMU}
 (also known as fractionally subadditive~\cite{BhawalkarR11WelfareGfCAwIB,Feige06OnMaximizingWwUFaS}) valuations.
A \XOS\ valuation represents a set of alternatives (e.g. travel destinations), where
the valuation for subsets of  goods (e.g. attractions) within each alternative
is additive. The value of any set of goods, e.g. dining and skiing, is the highest value for these goods among all alternatives.
That is, an \XOS\ valuation is essentially a depth-two tree with a MAX root over
SUM trees with goods as leaves.
 \XOS\ valuations are intuitive and very expressive:
they can represent any submodular valuation~\cite{LehmannLN02CombinatorialAwDMU} and can approximate any valuation in the subadditive superclass to a
$O(\log n)$ factor~\cite{BhawalkarR11WelfareGfCAwIB,Dobz}.
We also consider subclasses of submodular functions in the hierarchy, namely gross substitutes~\cite{CramtonSS06CombinatorialA,GulS99WalrasianEwGS} and \OXS~\cite{BuchfuhrerSS10ComputationaIiCPP,DayR06AssignmentPaCA, DobzinskiNS06TruthfulRMfCA,Singer10BudgetFM}
functions. 
Gross substitutes valuations are characterized by the lack of pairwise synergies among items:
for example, if the value of each of three items is the same, then no pair can have a strictly higher value than the other two pairs.
Finally, the \OXS\ class includes valuations representable as the SUM of MAX of item values. All these classes include linear valuations.

We also analyze the model of approximate learning everywhere with value queries, due to  Goemans et al.~\cite{nick09}.
In this model, the learner can adaptively pick a sequence of sets $S_{1}, S_{2}, \dots$ and  query the values $\targetf(S_{1}), \targetf(S_{2}), \dots$.
Unlike the high confidence and high accuracy requirements of PMAC, this model requires approximately learning $\targetf(\cdot)$ with certainty on all $2^{n}$ sets.
We provide upper and lower bounds in this model for the same valuation classes.

Finally, we introduce a more realistic variation of these models, in which the learner can obtain information only via prices. 
This variation is natural in settings where an agent with valuation $\targetf$ is interested in purchases of goods.

\paragraph{Our Results.}
We establish lower and upper bounds, the most general of them being almost tight, on the learnability of valuation classes
in the aforementioned hierarchy. 

\begin{enumerate}
\vspace{-0.25\baselineskip}
\item 
We show a nearly tight $O(\sqrt{n})$ upper bound  and $\Omega(\sqrt{n}/\log n)$ lower bound  on the learnability of \XOS\ valuations in the PMAC model.
The key element in our upper bound
is to show that any \XOS\ function can be
approximated by the square root of a linear function to within a factor $O(\sqrt{n})$.
Using this, we then reduce the problem of PMAC-learning \XOS\ valuations to the standard problem of learning linear separators in the PAC model which can be done via a number of efficient algorithms.
Our $\Omega(\sqrt{n}/\log n)$ lower bound is information theoretic, applying to any procedure that uses a polynomial number of samples.
 We also show an $O(\sqrt{n} \log n)$ upper bound on  the learnability of subadditive valuations in the PMAC model.


\vspace{-0.25\baselineskip}
\item 
We establish a target-dependent learnability result for \XOS\ functions. 
Namely, we show
the class of \XOS\ functions representable with at most $R$ trees can be PMAC-learned to an $O(R^{\eta})$ factor in time $n^{O(1/\eta)}$ for any $\eta>0$. In particular, for $R$ polynomial in $n$, we get learnability to an $O(n^{\eta})$  factor in time $n^{O(1/\eta)}$ for any $\eta\! >\! 0$. Technically, we prove this result   via a novel structural result showing that a \XOS\ function can be approximated well by the $L$-th root of a degree-$L$ polynomial over the natural feature representation of the set $S$.
Conceptually, this result highlights the importance of the complexity of the target function for polynomial time learning\footnote{Since the class of \XOS\ functions representable with at most a polynomial number of trees has small complexity, learnability would be immediate if we did not care about computational efficiency.}.

\vspace{-0.25\baselineskip}
\item 
 By exploiting novel structural results on approximability with simple functions, we provide much better
upper bounds for other interesting subclasses of \OXS\ and \XOS. These include \OXS\ and \XOS\ functions with
a small number of leaves per tree and \OXS\ functions with a small number of trees. Some of these classes have been considered
\mbox{in the context of economic optimization problems~\cite{secretary,BhawalkarR11WelfareGfCAwIB, BuchfuhrerDFKMPSSU10InapproximabilityfVBCA}, but we are the first to study their learnability.}
We also show that the previous $\tilde{\Omega}(n^{1/3})$ lower bound for PMAC-learning submodular functions~\cite{BalcanH10LearningSF} applies to the much simpler class of gross substitutes.

\vspace{-0.25\baselineskip}
\item 

The structural results we derive for analyzing learnability in the distributional learning setting also have implications for the model of exact learning with value queries~\cite{avrim03,nick09,SF}. In particular, they lead to new upper bounds for \XOS\ and \OXS\ as well as new lower bounds for \XOS, gross substitutes, and \OXS.

\vspace{-0.25\baselineskip}
\item 
Finally, we introduce a new model for learning with prices in which the learner receives less information on the values $\targetf(S_{1}), \targetf(S_{2}), \dots$: for each $l$, the learner can only \emph{quote} a price $p_{l}$ and observe whether the agent buys $S_{l}$ or not, i.e. whether $p_{l} \leq \targetf(S_{l})$ or not.
This model is more realistic in economic settings where agents 
interact with a seller via prices only.
Interestingly, many of our upper bounds, both for PMAC-learning and learning with value queries, are preserved in this model (all lower bounds automatically continue to hold).
\end{enumerate}
\vspace{-0.25\baselineskip}

Our results are summarized in Table~\ref{table:summary}. Note that all our upper bounds are efficient and all the lower bounds are information theoretic.
Our analysis has a number of interesting byproducts that should be of interest to the Combinatorial Optimization community.
 For example, it implies that recent lower bounds of~\cite{avrim03,GKTW09,nick09,SF} on
 optimization under submodular cost functions also apply to the smaller classes of \OXS\ and  gross substitutes.

\begin{table}[bht]
\begin{tabular}{|@{}c||@{}c|@{\:}c@{\:}||c@{\:}|c@{\:}||@{}}
\hline
Classes of valuations& PMAC~{\small\cite{BalcanH10LearningSF}}  & PMAC with prices
& Value Queries~{\small\cite{nick09}}  & VQ with prices
\\ \hline\hline
subadditive & $\tilde{\Theta}(n^{1/2})$~\thispaper & $\tilde{\Theta}(n^{1/2})$~\thispaper & $O(n)$~\folklore & $O(n)$~\thispaper
\\ \hline
\XOS\ & $\tilde{\Theta}(n^{1/2})$~\thispaper & $\tilde{\Theta}(n^{1/2})$~\thispaper & $\tilde{\Omega}(n^{1/2})$~\thispaper~{\small\cite{nick09}} & $\tilde{\Omega}(n^{1/2})$~\thispaper
\\
\XOS\ with $\leq R$ trees& $O(R^{\eps})$~\thispaper & $O(R^{\eps})$~\thispaper & $O(R)$~\thispaper & $O(R)$~\thispaper
\\ \hline
submodular & $(\tilde{\Omega}(n^{1/3}), O(n^{1/2}))$~{\small\cite{BalcanH10LearningSF}} &
   \begin{tabular}{@{}c@{}} $(\tilde{\Omega}(n^{1/3}), O(n^{1/2}))$ \\ \thispaper \end{tabular}
   & $\tilde{\Theta}(n^{1/2})$~~{\small\cite{nick09}} & --
\\ \hline
gross substitutes & $\tilde{\Omega}(n^{1/3})$~\thispaper & $\tilde{\Omega}(n^{1/3})$~\thispaper  & $\tilde{\Omega}(n^{1/2})$~\thispaper & $\tilde{\Omega}(n^{1/2})$~\thispaper
\\ \hline
\begin{tabular}{@{}c@{}} \OXS\ with $\leq R $ trees\\ or $\leq\! R$ leaves per tree\end{tabular}
& $O(R)$~\thispaper & $O(R)$~\thispaper & $O(R)$~\thispaper & $O(R)$~\thispaper \\
\hline
\end{tabular}
\caption{Lower and upper bounds for learnability factors achievable in different models for standard classes of valuations (presented in decreasing order of generality). All the upper bounds refer to polynomial time algorithms.
Our construction for the $\Omega(n^{1/2}/\log n) = \tilde{\Omega}(n^{1/2})$ lower bound on learning \XOS\ valuations with value queries is simpler than the construction for the same asymptotic lower bound of Goemans et al.~\cite{nick09}.}
\label{table:summary}
\end{table}

\paragraph{Related Work}
We study classes of valuations with fundamental properties (subadditivity and submodularity)
or that are natural constructs used widely for optimization in economic settings~\cite{book07}:
\XOS~\cite{BhawalkarR11WelfareGfCAwIB,DobzinskiNS05ApproximationAfCAwCFB,DobzinskiNS06TruthfulRMfCA, Feige06OnMaximizingWwUFaS,LehmannLN02CombinatorialAwDMU}, i.e. MAX of SUMs,
\OXS~\cite{BuchfuhrerDFKMPSSU10InapproximabilityfVBCA,BuchfuhrerSS10ComputationaIiCPP,DayR06AssignmentPaCA,Singer10BudgetFM}, i.e. SUM of MAXs,
and gross substitutes, fundamental in allocation problems~\cite{AusubelM02AscendingAwPB,CramtonSS06CombinatorialA,GulS99WalrasianEwGS}.

We focus on two  widely studied learning paradigms:
 approximate learnability in a distributional setting~\cite{AB99,KV:book94,Valiant:acm84,Vapnik:book98}
and approximate learning everywhere with value queries~\cite{avrim03,nick09,LahaieCP05MoreotPoDQiCALALaHI,SF}.
In the first paradigm, we use a model introduced by~\cite{BalcanH10LearningSF}
for the approximate learnability of submodular functions.
We circumvent the main negative result in~\cite{BalcanH10LearningSF} for certain interesting classes and match the main positive result in~\cite{BalcanH10LearningSF} for the more general \XOS\ and subadditive classes.
With a few recent exceptions~\cite{nick09,SF},
models for the value queries paradigm require exact learning
and are necessarily limited to much less general function classes than the ones we study here:
read-once and Toolbox DNF valuations~\cite{avrim03}, polynomial or linear-threshold valuations~\cite{LahaieP04ApplyingLAtPE} or  MAX or SUM (of bundles) valuations~\cite{LahaieCP05MoreotPoDQiCALALaHI}. The latter two works also consider demand queries,
where the learner can specify a set of prices and obtain a preferred bundle at these prices. In contrast, in our variation of learning with prices,
 the learner and the agent  focus on one price only (for the current bundle) instead of as many as $2^{n}$ prices.

\paragraph{Paper structure}
After defining valuation classes and our models  in Section~\ref{sec:model},
 we study the distributional learnability of valuation classes in decreasing order of generality.
First, Section~\ref{sec:XOS} presents our results on \XOS\ and subadditive valuations, including  our most general bounds, that are almost tight.
Section~\ref{sec:GS-OXS} presents a  hardness result  for gross substitutes and positive results  on several interesting subclasses of \OXS.
Section~\ref{value-queries} provides positive results
for most of these classes in learning with value queries.
Finally,
in Section~\ref{sec:prices} we show that many of our results
extend to a natural  framework\ in economic applications, even though the learner receives less information in this framework.

\vspace{-0.5\baselineskip}
\section{Preliminaries}

\label{sec:model}
\vspace{-0.25\baselineskip}

We consider a universe $[n]\!=\!\{1, \ldots, n\}$ of items and
\emph{valuations}, i.e. monotone non-negative set functions $f:2^{[n]} \!\to\! {\bR}_{+}$:
 $f(S \cup\{i\}) \geq f(S) \geq 0, \forall S \subseteq [n], \forall i\not\in S$.
For a set $S \subseteq [n]$ we denote by $\ind(S) \in \set{0,1}^{n}$ its indicator vector; so 
$(\ind(S))_{i}=1$ if $i \!\in\! S$ and $(\ind(S))_{i}=0$ if $i \!\not\in\! S$.
We often use this natural isomorphism between $\set{0,1}^n$ and $2^{[n]}$.

\smallskip\noindent
\textbf{Classes of Valuation Functions}.
\label{val-defs}
It is often the case that valuations are quite structured in terms of representation or constraints on the values of different sets.
We now define and give intuition for most valuation classes we focus on.

The following standard properties have natural interpretations in economic settings. A subadditive valuation models the lack of synergies among sets: a set's  value is at most the sum of the values of its parts.
A submodular valuation models decreasing marginal returns: an item $j$'s marginal value cannot go up if one expands the base set $S$ by item $i$.
\begin{Definition}
\label{subadditive}
A valuation $f :2^{[n]} \to \bR_+$ is called \emph{subadditive} if and only if $ f(S \cup S') \leq f(S) + f(S'),  \forall S, S' \subseteq [n]$.
\label{submodular}
$\!{}\!$A valuation $f$
is called \emph{submodular} if and only if
$f(S \cup\{i,j\}) - f(S \cup\{i\}) \leq f(S\cup\{j\}) - f(S)$,
 $\forall S \subseteq [n], \:\forall i,j \not\in S$.
\end{Definition}

\XOS\ is an important class of subadditive, but not necessarily submodular, valuations studied in combinatorial auctions~\cite{DobzinskiNS05ApproximationAfCAwCFB,DobzinskiNS06TruthfulRMfCA, Feige06OnMaximizingWwUFaS,LehmannLN02CombinatorialAwDMU}.
A  valuation  is \XOS\ if and only if it can be represented as a depth-two tree with a MAX root and SUM inner nodes. Each such SUM node has as leaves a subset of items with associated positive weights.
For example, a traveler may choose the destination of maximum value among several different locations, where each location has a number of amenities and the valuation for a location is linear in the set of amenities.
\begin{Definition}  
A valuation $f$ is \XOS\ if and only if it can be represented as
 the maximum of $k$ linear valuations, 
for some $k\! \geq\! 1$.
That is, $f(S) = \max_{j=1\dots k} w_{j}\transpose \ind(S)$ where $w_{ji} \geq 0, \forall j=1\dots k, \forall i = 1\dots n$.
\label{def:XOS}
\end{Definition}
We say that item $i$ appears as a leaf in a SUM tree $j$ if $i$ has a positive value in tree $j$.

As already mentioned, any submodular valuation can be expressed as \XOS\footnote{As showin in ~\cite{LehmannLN02CombinatorialAwDMU}, any submodular $f$  can be represented as the MAX of $n!$ \OR\ trees, each with $n$ leaves:
for every permutation $\pi$ of $[n]$, we build a \OR\ tree $T_{\pi}$
 with one leaf for each item $j \in [n]$,
where item $\pi(j)$ has weight its marginal value
$
 f(\{ \pi(1),\dots, \pi(j-1), \pi(j) \}) ~ - ~ f(\{ \pi(1),\dots,\pi(j-1) \}) $.
}.
When reversing the roles of operators MAX and SUM we obtain a strict \emph{sub}class of submodular valuations, called \OXS,\footnote{\XOS\ and \OXS\ stand for XOR-of-OR-of-Singletons and OR-of-XOR-of-Singletons, where
MAX is  denoted by XOR and SUM by OR\cite{noam-book-chapter,Sandholm99AlgorithmfOWDiCA}.}
that is also relevant to auctions~\cite{DayR06AssignmentPaCA, DobzinskiNS06TruthfulRMfCA,LehmannLN02CombinatorialAwDMU,Singer10BudgetFM}.
To define \OXS\ we also define a unit-demand valuation,
in which the value of any  set $S$ is the highest weight of any item in $S$.
A unit-demand valuation is essentially a tree, with a MAX root
and one leaf for each item with non-zero associated weight.
In an \OXS\ valuation, a set's value is given by the best way to split the set among several unit-demand valuations.
An \OXS\ valuation $f$ has a natural representation as a depth-two tree, with a SUM node at the root (on level 0), and subtrees\footnote{Another \OXS\ encoding uses a weighted bipartite graph $G_{n,k}$ 
where edge $(i,j)$ has the weight of item $i$ in $f_{j}$;
$f(S)$ is the weight of a maximum matching of $S$ to the $k$ nodes for the unit-demand $f_{j}$'s.
Also, \OXS\ valuations with weights $\{0,\! 1\}$ are exactly rank functions of transversal matroids.} corresponding to the unit-demand valuations $f_{1}, \dots, f_{k}$.
The value $f(S)$ of any set $S$ corresponds to best way of
 partitioning $S$ into $(S_{1}, \dots, S_{k})$ and adding up the per-tree values
 $\{f_{1}(S_{1}), \dots, f_{k}(S_{k})\}$.
\begin{Definition}
$\!{}$A \emph{unit-demand valuation} $f$ is given by weights $\{w_{1},..., w_{n}\!\} \!\subset\! \bR_{+}$ such that
$f(S) = \max_{i \in S} w_{i}, \forall S\! \subseteq\! [n]$.
An \emph{\OXS\ valuation} $f$ is given by the convolution of $k \geq 1$ unit-demand valuations $f_{1}, \dots, f_{k}$:
that is,
\\ \noindent
$f(S) = \max \{ f_{1}(S_{1}) + \dots + f_{k}(S_{k}): (S_{1}, \dots, S_{k}) \text{ is a } \text{partition of } S \}, \forall S\! \subseteq\! [n]$.
\label{def:OXS}
\end{Definition}

Finally, we consider \emph{gross substitutes} (GS) valuations,
of great interest in allocation problems~\cite{CramtonSS06CombinatorialA,GulS99WalrasianEwGS,book07}.
Informally, an agent with a gross substitutes valuation would not buy fewer items of one type (e.g. skis) if items of another type (e.g. snowboards) became more expensive.
That is, items can be substituted one for another in a certain sense, which is not the case
for, e.g. skis and ski boots.
See Section~\ref{sec:GS} for a formal definition and more detailed discussion.

As already mentioned, the classes of valuations we reviewed thus far form a strict hierarchy. (See~\cite{LehmannLN02CombinatorialAwDMU} for  examples  separating these classes of valuations.)
\begin{Lemma}~\cite{LehmannLN02CombinatorialAwDMU}
 $\OXS\ \subsetneq \text{gross\ substitutes} \subsetneq \text{submodular} \subsetneq \XOS\ \subsetneq \text{subadditive}$.
\label{stmt:manyIncl}
\end{Lemma}
Only the class of submodular valuations has been studied from an approximate learning perspective~\cite{BalcanH10LearningSF,nick09}.
We study the approximate  learnability of all other classes in this hierarchy, in a few natural models that we~introduce~now.

\smallskip\noindent
\textbf{Distributional Learning: PMAC}.
We primarily study learning in the PMAC model of ~\cite{BalcanH10LearningSF}.
We assume that the input for a learning algorithm is  a set $\trainS$
of polynomially many labeled examples drawn i.i.d.~from
some fixed, but unknown, distribution $D$ over points in $2^{[n]}$.
The points are labeled by a fixed, but unknown, target function $\targetf : 2^{[n]} \rightarrow \bR_+$.
The goal is to output a hypothesis function $f$ such that,
with high probability over the choice of examples,
the set of points for which $f$ is a good approximation for $\targetf$
has large measure with respect to $D$. Formally:

\begin{Definition}~ 
We say that a family $\cF$ of valuations is  \emph{PMAC-learnable with approximation factor $\alpha$} if there exists an
algorithm $\cA$ such that for any distribution $D$ over $2^{[n]}$, for any
target function $\targetf \in \cF$, and for any sufficiently small $\eps \geq 0, \delta \geq 0$,   $\cA$ takes as input  samples $\{(S_i, \targetf(S_i))\}_{1\leq i \leq m}$ where each $S_i$ is drawn independently
from  $D$ and  outputs a valuation $f : 2^{[n]} \to \bR$ 
such that
$
 \probover{ S_1, \dots, S_m \sim D } {
\probover {S \sim D} {f(S) \leq \targetf(S) \leq \alpha f(S)} \geq 1-\eps } \geq 1-\delta
$.
$\cA$ must use $m = \poly(n, \frac{1}{\eps},\frac{1}{\delta})$ samples and must have
running time  $\poly(n, \frac{1}{\eps},\frac{1}{\delta})$. 
\label{def:PMAC}
\end{Definition}
\noindent
PMAC stands for Probably Mostly Approximately Correct (the PAC model~\cite{Valiant:acm84} is a special case of PMAC with $\alpha\!=\!1$).

\smallskip\noindent
\textbf{Learning with Value Queries}.
We also consider the model of learnability everywhere (in the same approximate sense) with value queries.
In this model, the learning algorithm is allowed to query the value of the unknown target function $\targetf$ on a polynomial number of sets $S_{1}, S_{2}, \dots$, that may be chosen in an adaptive fashion.
The algorithm must then output in polynomial time a function $f$ that approximates $\targetf$ everywhere,
namely $f(S) \leq \targetf(S) \leq \alpha f(S), \,\forall S \subseteq [n]$.
A formal definition of this model and the results are presented in Section~\ref{value-queries}.

\smallskip\noindent
\textbf{Learning with Prices}. This framework aims to model economic interactions more realistically and considers a setting where an agent with the target valuation $\targetf$ is interested in purchasing bundles of goods.
In this framework, the learner does not obtain the value of $\targetf$ on each input set $S_{1}, S_{2}, \dots$.
Instead, for each input set $S_{l}$ the learner quotes a price $p_{l}$ on $S_{l}$ and observes whether the agent purchases $S_{l}$ or not, i.e. whether $p_{l} \leq \targetf(S_{l})$ or not.
 The goal remains to approximate the function $\targetf$ well, i.e. within an $\alpha$ multiplicative factor: 
   on most sets from $D$ with high confidence for PMAC-learning
   and on all sets with certainty for learning everywhere with value queries.
This  framework\ and the associated results  are presented in Section~\ref{sec:prices}.

\section{PMAC-learnability of \XOS\ valuations and subadditive valuations}
\label{sec:XOS}
In this section we give nearly tight lower and upper bounds of $\tilde{\Theta}(\sqrt{n})$ for the PMAC-learnability of \XOS\ and subadditive valuations.
In contrast, there is a $\tilde{\Theta}(n^{1/6})$ gap between the existing bounds for submodular valuations~\cite{BalcanH10LearningSF}.
Furthermore, we reveal the importance of considering the complexity of the target function (in a natural representation) for polynomial-time PMAC learning.
We show that \XOS\ valuations representable with a polynomial number of \OR trees are
PMAC-learnable to a $n^{\eta}$ factor in time $n^{1/\eta}$, for any  $\eta > 0$.
Finally, we show that \XOS\ valuations representable with an arbitrary number of \OR trees, each with at most $R$ leaves, are PMAC-learnable to an $R$ factor.

\subsection{Nearly tight lower and upper bounds for learning \XOS\ and subadditive functions}
\label{upper bound of XOS}
\label{lower bound of XOS}
We establish our $\tilde{\Theta}(\sqrt{n})$ bounds by showing an
$\Omega(\sqrt{n}/{\log n})$ lower bound for the class of \XOS\ valuations (hence valid for subadditive valuations) and upper bounds of
$O(\sqrt{n} )$ and $O(\sqrt{n} \log n)$ for the classes of \XOS\ and subadditive valuations respectively.
We note that our lower bound construction is much simpler and gives a better bound
than the $\Omega(n^{1/3}/\log n)$ construction of~\cite{BalcanH10LearningSF}.
However, the latter construction is for matroid rank functions, a significantly smaller class.
For our upper bounds we provide structural results
showing that \XOS\ and subadditive functions can be approximated by a linear function to an $O(\sqrt{n})$ and $O(\sqrt{n} \ln n)$ factor respectively. We can then PMAC-learn these classes via a reduction to the classical problem of PAC-learning a linear separator.

\begin{Theorem}
\label{stmt:XOS-upper-bound}
\label{stmt:subadd-upper-bound}
\label{alg-lin-sep-first}
\label{stmt:XOS-lower-bound}
The classes of $\XOS$ and subadditive functions are PMAC-learnable to a
$\tilde{\Theta}(\sqrt{n})$ approximation factor.%
\end{Theorem}%

\renewcommand{\targetfpower}[1]{{(\targetf(#1))^{2}}}

\vspace{-0.5\baselineskip}
\noindent\begin{proofsketch} 
\textbf{Lower bound}:
We start with an information theoretic lower bound showing that the class of \XOS\ valuations cannot be learned with an approximation factor of $o(\frac{\sqrt{n}}{\log n})$ from a polynomial number of samples.

Let  $k\!=\!n^{\frac{1}{3}\log\log n}\!$.
For large enough $n$ we can show that there exist sets $A_1,A_2,...,A_k \subseteq [n]$ such that

\noindent
(i) $\sqrt{n}/2\le |A_i|\le 2\sqrt{n}$ for any $1\le i\le k$, i.e. all sets have large size $\Theta(\sqrt{n})$ and 
\\
(ii) $|A_i\cap A_j|\le \log n$ for any $1\le i<j\le k$, i.e. all pairwise intersections have small size $O(\log n)$. 

We achieve this via a simple probabilistic argument where we construct each $A_i$ by picking each element in $[n]$ with probability $\frac{1}{\sqrt{n}}$. Let random variables $Y_i\!=\!|A_i|$ and $X_{ij} \!=\! |A_i \!\cap\! A_j|$. 
Obviously, $E[Y_i]\!=\!\sqrt{n}$ and $E[X_{ij}] \!=\!1$. By Chernoff bounds, 
$$
\prob{\sqrt{n}/2<Y_i<2\sqrt{n}}  > 1- 2 e^{-\sqrt{n}/8}
\text{ ~ and ~ }
\textstyle \prob{X_{ij} > \ln n} < \frac{e^{\ln n}}{\ln n^{\ln n}} = n^{-(\ln\ln n - 1)},
~~\forall 1\le i<j\le k.
$$
By union bound the probability that (i) and (ii) hold is at least $1- 2\, k\, e^{-\sqrt{n}/8} - k^2 n^{-(\ln\ln n-1)} > 0$.

Given the existence of  the \setfamily\ $\mathcal{A} = \{A_{1}, \dots, A_{k}\}$ of sets with properties (i) and (ii) above, we construct a hard family of \XOS\ functions as follows.
For any sub\setfamily\ $\mathcal{B} \subseteq \mathcal{A}$, we construct an \XOS\ function
$f_{\mathcal{B}}$
with large values for sets $A_{i} \in \mathcal{B}$ and small values for sets $A_{i} \not\in \mathcal{B}$.
Let $h_{A_i}(S)=|S\cap A_i|$ for any $S\subseteq [n]$.
For any sub\setfamily\ $\mathcal{B}\subseteq \mathcal{A}$,
define the \XOS\ function $f_{\mathcal{B}}$ by $f_{\mathcal{B}}(S)=\XOR_{A_i\in \mathcal{B}} \ h_{A_i}(S)$.
We claim that
$f_{\mathcal{B}}(A_i) = \Omega(\sqrt{n})$, if $A_i\in \mathcal{B}$ but
$f_{\mathcal{B}}(A_i) = O(\log n)$, if $A_i\not\in \mathcal{B}$.
Indeed, for any $A_i\in \mathcal{B}$, we have $h_{A_i}(A_{i})=|A_i|\ge \sqrt{n}/2$, hence $f_{\mathcal{B}}(A_i)=\Omega(\sqrt{n})$; for any $A_j\not\in \mathcal{B}$, by our construction of $\mathcal{A}$, we have $h_{A_i}(A_j)=|A_i\cap A_j|\le \log n$, implying $f_{\mathcal{B}}(A_j)=O(\log n)$.
For an unknown $\mathcal{B}$,
the problem of learning  $f_{\mathcal{B}}$
  within a factor of $o(\sqrt{n}/\log n)$
  under a uniform distribution on $\mathcal{A}$
  amounts to distinguishing $\mathcal{B}$ from $\mathcal{A}$. This is not possible from a polynomial number of samples since $|\mathcal{A}|=n^{\frac{1}{3}\log\log n}$.
 In particular, if $\mathcal{B} \subseteq \mathcal{A}$ is chosen at random,
then any algorithm from a polynomial-sized sample will have error
$\Omega(\frac{\sqrt{n}}{\log n})$ on a region of probability mass greater than
$\frac{1}{2} - \frac{1}{\poly(n)}$.

\smallskip\noindent
\textbf{Upper bounds}:
We show that the class of \XOS\ valuations can be PMAC-learned to a $O(\sqrt{n})$ factor
and that the class of subadditive valuations  can be PMAC-learned to a $O(\sqrt{n} \log n)$ factor, by using  $O(\frac{n}{\epsilon} \log \frac{n}{\delta \epsilon})$ training examples and running time  $\poly(n, \frac{1}{\eps},\frac{1}{\delta})$.
To prove these bounds we start by providing a structural result (Claim~\ref{Satoru} below) showing that \XOS\ valuations can be approximated to a $\sqrt{n}$ factor by the square root of a linear function.

\begin{Claim}
Let $f: 2^{[n]} \rightarrow \reals_{+}$ be a non-negative \XOS\ function with
$f(\emptyset)=0$.
Then there exists a function $\hf$ of the form $\hf(S) = \sqrt{ w \transpose \ind(S) }$
where $w \in \bR^n_+$
such that $\hf(S) \leq f(S) \leq \sqrt{n} \hf(S)$ for all $S \subseteq [n]$.
\label{Satoru}
\end{Claim}
\begin{proof}
\XOS\ valuations are known~\cite{Feige06OnMaximizingWwUFaS} to be equivalent
to fractionally subadditive valuations.
A function $f:2^{[n]}\to \reals$ is called \emph{fractionally subadditive} if $f(T) \leq \sum_{S} \lambda_S f(S)$ whenever
$\lambda_S \geq 0$ and $\sum_{S:s \in S} {\lambda_S} \geq 1$ for any $s\in T$.

We can show the following property of \XOS\ valuations: 
for any \XOS\ $f$ we have $f(T)=\max{ \{\sum_{i\in T}x_i| x \in  P(f)\}}$, where $P(f)$ is the associated polyhedron $\{x\in \mathbf{R}_{+}^{n}:\sum_{i\in S}x_i\le f(S),\forall S\subseteq [n]\}$.
Informally, this result states that one recovers $f(T)$ when optimizing in the direction given by $T$ over the polyhedron $P(f)$ associated with $f$.
The proof of this result involves a pair of dual linear programs, one corresponding to the maximization and another one that is tailored for fractional subadditivity, with an optimal objective value of  $f(T)$.
Formally,
for any $T\subseteq [n]$ we have $\sum_{i\in T}x_i\le f(T)$ for any $x\in P(f)$. Therefore $f(T)\ge \max{ \{\sum_{i\in T}x_i| x \in  P(f)\}}$. Now we prove that in fact $$f(T)\le \max{ \{\sum_{i\in T}x_i| x \in  P(f)\}}.$$
Consider the linear programming (LP1) for the quantity $\max{\{x(T)| x \in  P(f)\}}$ and its dual (LP2): we assign a dual variable $y_S$ for each constraint in (LP1), and we have a constraint corresponding to each primal variable indicating that the total amount of dual corresponding to a primal variable should not exceed its coefficient in the primal objective.

\begin{minipage}[h]{0.45\textwidth}
\begin{eqnarray}
    \max\sum_{i\in T}x_i & & \mathrm{(LP1)}\nonumber \\
   s.t. \sum_{i\in S}x_i \le f(S) & & \forall S \subseteq [n],\nonumber \\
    x_i \geq 0 & & \forall i \in [n].\nonumber
\end{eqnarray}
\end{minipage}
\begin{minipage}[h]{0.45\textwidth}
\begin{eqnarray}
    \min \sum_{S\subseteq [n]}y_Sf(S) &&  \mathrm{(LP2)}\nonumber \\
     s.t. \sum_{S:i \in S}y_S \ge 1 & & \forall i \in T, \nonumber \\
     y_S \geq 0 & & \forall S \subseteq [n]. \nonumber
\end{eqnarray}
\end{minipage}

\medskip
The classical theory of linear optimization gives that the optimal primal solution equals the optimal dual solution. 
Let $y^*$ be an optimal solution of (LP2). Therefore $$\sum_{S\subseteq [n]}y^*_Sf(S)=\max{\{x(T)| x \in  P(f)\}}.$$ Since $f$ is fractionally subadditive and $\sum_{S:i \in S}y^*_S \ge 1, \forall i \in T$, we have $f(T)\le \sum_{S\subseteq [n]}y^*_Sf(S),$ hence $f(T)\le \max{\{x(T)| x \in  P(f)\}}.$ This completes the proof of the fact that $f(T)=\max{ \{\sum_{i\in T}x_i| x \in  P(f)\}}$.

Given this result, 
we proceed as follows
(a very similar approach is used by~\cite{nick09} for submodular functions).
Define 
$P=\{x\in \mathbf{R}^{n}: (|x_1|,...,|x_n|) \in P(f)\}$. Since $P$ is bounded and \emph{central symmetric} (i.e. $x\!\in\! P \Leftrightarrow -x\!\in\! P$),  there exists~\cite{John} an ellipsoid $\mathcal{E}$ containing $P$ such that $\frac{1}{\sqrt{n}}\mathcal{E}$ is contained in $P$. Hence for $\hf(T)=\max\{\sum_{i\in T}x_i:x \!\in\! \frac{1}{\sqrt{n}}\mathcal{E}\}$, we have $\hf(T) \!\le\! f(T) \!\le\! \sqrt{n}\hf(T), \forall T\subseteq [n]$. At last, basic calculus implies $\hf(T)=\sqrt{ w \transpose \ind(T) }$ for some  $w \in \bR^n_+$.
\end{proof}

For PMAC-learning \XOS\ valuations to with an approximation factor of $\sqrt{n+\eps}$,
we apply Algorithm~\ref{alg-lin-sep} with parameters $R=n$, $\epsilon$, and $p=2$. 
The proof of correctness of  Algorithm~\ref{alg-lin-sep} follows by using the structural result in Claim~\ref{Satoru} and a technique of \cite{BalcanH10LearningSF} that we sketch briefly here. Full details of this proof appear in Appendix~\ref{appendix-lin-sep}.

Assume first that $\targetf(S) > 0 $ for all $S \neq \emptyset$.
The key idea is that  Claim~\ref{Satoru}'s structural result implies that the following examples in $\bR^{n+1}$ are linearly separable since
$n w\transpose\ind(S) - \targetfpower{S} \geq 0$ and  $n w\transpose\ind(S) - (n+\epsilon) \targetfpower{S} < 0$.
\newcommand{\exap}{\mathrm{ex}^+_S}
\newcommand{\exam}{\mathrm{ex}^-_S}
$$
\begin{array}{lll}
\text{Examples labeled $+1$:}& ~~~ \exap := (\ind(S), \targetfpower{S}) &\quad\forall S \subseteq [n] \\
\text{Examples labeled $-1$:}& ~~~ \exam := (\ind(S), (n+\epsilon) \cdot \targetfpower{S}) &\quad\forall S \subseteq [n]
\end{array}
$$
This suggests trying to reduce our learning problem to the standard problem
of learning a linear separator for these examples in the standard PAC model~\cite{KV:book94,Vapnik:book98}.
However, in order to apply standard techniques to learn such a linear separator,
we must ensure that our training examples are i.i.d.
To achieve this, we create a i.i.d.\ distribution $D'$ in $\bR^{n+1}$
that is related to the original distribution $D$ as follows.
First, we draw a sample $S \subseteq [n]$ from the distribution $D$ and then flip a fair coin for each.
The sample from $D'$ is labeled $\exap$ i.e. $+1$ if the coin is heads and $\exam$ i.e.  $-1$ if the coin is tails.
As mentioned above, these labeled examples are linearly separable in $\bR^{n+1}$.
Conversely, suppose we can find a linear separator that classifies most of the examples coming from
$D'$ correctly.
Assume that this linear separator in $\bR^{n+1}$
is defined by the function $u \transpose x = 0$,
where $u = (\hat{w},-z)$, $w \in \bR^n$ and $z > 0$.
The key observation is that the function $f(S) = \frac{1}{(n+\epsilon)z} \hat{w} \transpose \ind(S)$
approximates $\targetfpower{\cdot}$ to within a factor $n+\epsilon$ on most of the points coming from $D$.

If $\targetf$ is zero 
 on non-empty sets,
then we can learn its set  $\cZ = \setst{ S }{ \targetf(S) = 0 }$ of zeros quickly since
$\cZ$ is closed to union and taking subsets for any subadditive $\targetf$.
In particular, suppose that there is at least an $\epsilon$ chance that a new example is a zero of $\targetf$,
but does not lie in the null subcube over the sample.
Then such a example should be seen in the next sequence of $\log(1/\delta)/\epsilon$ examples,
with probability at least $1-\delta$.
This new example increases the dimension of the null subcube by at least one,
and therefore this can happen at most $n$ times.

To establish learnability for the class of subadditive valuations, we note that any subadditive valuation can be approximated by an \XOS\ valuation to a
$\ln n$ factor~\cite{Dobz,BhawalkarR11WelfareGfCAwIB}
    \footnote{We are grateful to Shahar Dobzinski and Kshipra Bhawalkar for pointing out this fact to us.} and so, by Claim~\ref{Satoru}, any subadditive valuation is approximated to a $\sqrt{n} \ln n$ factor by a linear function.
This then implies that we can use Algorithm~\ref{alg-lin-sep} with parameters $R = n \ln^{2} n$, $\epsilon$, and $p = 2$.
Correctness then follows by a reasoning similar to the one for \XOS\ functions.
\end{proofsketch} 

\vspace{-0.2cm}
\begin{algorithm}[th]
{\bf Input:} Parameters: $R$, $\epsilon$ and $p$. Training examples $\cS= \!\left\{(S_1,\targetf(S_1)),
\ldots, (S_{\ell},\targetf(S_{\ell})) \right\}$.
\vspace{-0.2cm}
\begin{itemize}
\setlength{\parsep}{0pt}
\setlength{\itemsep}{0pt}
\setlength{\leftmargin}{0pt}
\item Let $\cSnz=\left\{(A_i,\targetf(A_i)) \!\in\! \cS\!: \targetf(A_i) \!\neq\! 0 \right\} \!\subseteq\! \cS$
the examples with non-zero values, \ensuremath{\cSz=\cS \setminus \cSnz} and
$$\cU_{0} = \cup_{{l \leq \ell; \targetf(S_l)=0}} {S_l}.$$

\item For each $i$ in $\{1, \dots, |\cSnz|\}$
let $y_i$ be the outcome of independently flipping a fair $\set{+1,-1}$-valued coin.

Let $x_i \in \bR^{n+1}$ be the point defined by
$
x_i ~=~ \begin{cases}
(\ \ind(A_i), (\targetf(A_i))^p\ ) &\quad\text{(if $y_i=+1$)} \\
(\ \ind(A_i), (R+\epsilon) \cdot (\targetf(A_i))^p\ ) &\quad\text{(if $y_i=-1$)}.
\end{cases}
$
\item  Find a linear separator $u = (\hat{w},-z) \in \bR^{n+1}$, where $\hat{w} \in \bR^n$ and $z > 0$,
such that $(x,\sgn(u\transpose x))$ is consistent with the labeled examples $(x_i,y_i) ~\forall i \in \{1, \dots, |\cSnz|\}$,
and with the additional constraint that $\hat{w}_j=0 ~\forall j \in \cU_{0}$.
\end{itemize}
\vspace{-0.2cm}
{\bf Output:} The function $f$ defined as $f(S) = \left(\frac{1}{(R+\epsilon)z}\: \hat{w} \transpose \ind(S)\right)^{1/p}$.
\caption{Algorithm for PMAC-learning  via a reduction to a binary linear separator problem.
}
\label{alg-lin-sep}
\end{algorithm}

\subsection{Better learnability results for \XOS\ valuations with polynomial complexity}
\label{sec:XOS-Poly-Trees}
In this section we consider the learnability of \XOS\ valuations representable with a polynomial number of trees.
Since this class has small complexity, it is easy to see that it is learnable in principle from a small sample size if we did not care about computational complexity.
Interestingly we can show that we can achieve good PMAC learnability via polynomial time algorithms. In particular, we show that \XOS\ functions representable with at most $R$ \OR trees can be PMAC-learned with a $R^{\eta}$ approximation factor in time $n^{O(1/\eta)}$, for any  $\eta > 0$.
This improves the approximation factor of Theorem~\ref{stmt:XOS-upper-bound} for all such \XOS\ functions.
Moreover, this implies that
\XOS\ valuations representable with a polynomial number of trees can be PMAC-learned within a factor of $n^{\eta}$, in time $n^{O(1/\eta)}$, for any $\eta > 0$.

\begin{Theorem}
For any $\eta>0$, the class of \XOS\ functions representable with at most $R = n^{O(1)}$ \OR trees
is PMAC-learnable in time $n^{O(1/\eta)}$ with approximation factor of $(R+\eps)^{\eta}$ by using
$O\left(\frac{n^{1/\eta}}{\epsilon} \left[\frac{\log(n)}{\eta} + \log\left(\frac{1}{\delta \epsilon}\right)\right] \right)$ training examples.
    \label{stmt:nToEpsXOSPMAC}
\end{Theorem}

\begin{proof}
Let $L =  1/\eta$ and assume for simplicity that it is integer. We start by deriving a key structural result. We show that \XOS\ functions can be
approximated well by the $L$-th root of a degree-$L$ polynomial over $(\chi(S))_i$ for  $i
\in [n]$.
Let $T_{1}, \dots, T_{R}$ be the $R$ \OR trees in an \XOS\ representation $\cT$ of $\targetf$.
For a tree $j$ and a leaf in $T_{j}$ corresponding to an element $i \in [n]$,
let $w_{ji}$ the weight of the leaf.
For any set $S$, let $k_{j}(S) = \sum_{i \in T_{j} \cap S} w_{ji} =
w_{j}^{T} \ind(S)$ be the sum of weights in  tree $T_{j}$ corresponding to leaves in $S$.
$k_{j}(S)$ is the value assigned to set $S$ by tree $T_{j}$.
Note that $\targetf(S) = \max_{j} k_{j}(S)$, i.e. the maximum value of any tree, from the definition of \XOR.
We define valuation $f'$ that averages the $L$-th powers of the values of all trees:
~$
f'(S) = 1/R\sum\nolimits_{j} k_{j}^{L}(S), \:\forall S \subseteq [n].
$~
We claim that $f'(\cdot)$ approximates $(\targetf(\cdot))^{L}$ to within an $R$ factor on all sets $S$, namely
\begin{eqnarray}
f'(S) \leq (\targetf(S))^{L} \leq R f'(S), ~\forall S \subseteq [n]
\label{eq:XOSApproxTargetf'}
\quad \text{i.e. }  \quad
\textstyle 1/R \sum\nolimits_{j} k_{j}^{L}(S) \leq \max\nolimits_{j} k_{j}^{L}(S) \leq \sum\nolimits_{j} k_{j}^{L}(S), ~\forall S \subseteq [n]
\label{eq:XOSApproxTargetf'kj}
\end{eqnarray}
The left-hand side inequalities in Eq.~\eqref{eq:XOSApproxTargetf'kj} follow as $\targetf$ has at most $R$ trees  and $k_{j'}^{L} (S) \!\leq\! \max_{j} k_{j}^{L} (S)$ for any tree $T_{j'}.$ The right-hand side inequalities in Eq.~\eqref{eq:XOSApproxTargetf'kj} follow immediately.

This structural result suggests re-representing each set $S$ by a new set of $\Theta(n^L)$ features, with one feature for each subset of $[n]$ with at most $L$ items. Formally, for any set $S \subseteq [n]$,
we denote by $\ind_M(S)$ its feature representation over this new set of features. $\ind_M(S)_{i_1, i_2, \ldots, i_L}=1$ if all items  $i_1, i_2, \ldots i_L$ appear in $S$ and $\ind_M(S)_{i_1, i_2, \ldots, i_L}=0$ otherwise.
It is easy to see that $f'$ is representable as a linear function over this new set of features. This holds for each $k_{j}^{L}(S) = (w_{j}^{T} \ind(S))^{L}$ due to  its multinomial expansion, that contains one term for each set of up to $L$ items appearing in tree $T_{j}$, i.e. for each such feature. Furthermore, $f'$ remains linear when the terms for each tree $T_{j}$ are added.

Given this, we can now use a variant of Algorithm~\ref{alg-lin-sep} with parameters $R$, $\epsilon$, and $p \!=\! L$ and to prove correctness we can use a reasoning similar to the one in Theorem~\ref{alg-lin-sep-first}.
Any sample $S_{l}$ is fed into Algorithm~\ref{alg-lin-sep}
as $(\ind_M(S_{l}), (\targetf(S_{l}))^{L})$ or
$( \ind_M(S_{l}), (R+\epsilon) \cdot(\targetf(S_{l}))^{L})$ respectively.
Since $f'$ is linear over the set of features, Algorithm~\ref{alg-lin-sep}
outputs
with probability at least $1-\delta$ a hypothesis $f''$ that approximates $\targetf$ to an $(R+\eps)^{1/L}$ factor on
any point  $\ind_M(S)$ corresponding to sets $S \subseteq [n]$ from a collection $\mathcal{S}$ with at  least an $1-\eps$ measure in $D$, i.e.
 $f''(\ind_M(S)) \leq \targetf(S) \leq (R+\eps)^{1/L} f''(\ind_M(S))$. 
We can output then hypothesis $f(S) = f''(\ind_M(S)), \forall S \subseteq [n]$,   defined on the initial ground set $[n]$ of items, that  approximates $\targetf(\cdot)$ well, i.e. for any $S \in \mathcal{S}$ we have
\[
f(S) = f''(\ind_M(S)) \leq
\targetf(S) \leq
(R+\eps)^{1/L} f''(\ind_M(S)) = (R+\eps)^{1/L} f(S)
\]
As desired, with high confidence the hypothesis $f$ approximates $\targetf$ to a $(R+\eps)^{\eta}$ factor on most sets from $D$.
\end{proof}

This result has an appealing interpretation in terms of representations of submodular functions.
We know that any submodular function is representable as an \XOS\ tree.
What Theorem~\ref{stmt:nToEpsXOSPMAC} implies is that (submodular) functions that are succinctly representable as \XOS\ trees can be PMAC-learned well.
Theorem~\ref{stmt:nToEpsXOSPMAC} is thus a target-dependent learnability result, in that the extent of learnability of a function depends on the function's complexity.
\label{sec:XOS-R-Leaves}

\subsection{Better learnability results for \XOS\ valuations with small \OR trees}
In this section we consider the learnability of another interesting subclass of \XOS\ valuations, namely \XOS\ valuations representable with ``small'' \OR trees and show learnability to a better factor than that in Theorem~\ref{stmt:XOS-upper-bound}. For example, consider a traveler deciding between many trips, each to a different location with a small number of tourist attractions.
The traveler has an additive value for several attractions at the same location.
This valuation can be represented as an \XOSLong\  function where each \OR tree  stands for a location and has a small number of leaves.
We now show good  PMAC-learning  guarantees for classes of functions of this type.

\vspace{-0.5\baselineskip}
\begin{Theorem}
For any $\eta>0$,
the class of \XOS\ functions representable with \OR trees with at most $R$ leaves
is properly PMAC-learnable with approximation factor of $R(1+\eta)$ by using
$m=O(\frac{1}{\epsilon} \left( n \log \log_{1+\eta} (\frac{H}{h})
+ \log(1/\delta) \right))$ and running time polynomial in $m$,
where $h$ and $H$ are the smallest and the largest non-zero values our functions can take.
\label{stmt:XOS-R-Leaves}
\end{Theorem}
\begin{proof}
 We show that the unit-demand 
 hypothesis $f$ output by
Algorithm~\ref{alg-XOR} produces the desired result.  The algorithm constructs a unit demand hypothesis function $f$ as follows. For any $i$ that appears in at least one set $S_j$ in the sample we define $f(i)$ as the smallest value $\targetf(S_j)$ over all the sets $S_j$ in the sample containing $i$. For $i$ that does not appear in any set $S_j$ define $f(i) = 0$. 

We start by proving a key structural result showing that $f$ approximates the target
function multiplicatively within a factor of $R$ over the sample.
That means:
\begin{eqnarray} f(S_l) \leq \targetf(S_l) \leq R
f(S_l)~~~\mathrm{~~for~~all~~} l \in \{1,2, \ldots,
m\}.\label{approx-sample-RORPMAC}\end{eqnarray}
To see this note that for any $i \in S_l$ we have $\targetf(i) \leq
\targetf(S_l)$, for $l \in \{1,2, \ldots, m\}$.
 So \begin{eqnarray}f(i) \geq \targetf(i) ~~\mathrm{~for~any~} i \in
S_1 \cup \ldots \cup S_m. \label{RORPMAC:one}\end{eqnarray}
Therefore for any $l \in \{1,2, \ldots, m\}$. :
$$\targetf(S_l) \leq R \max_{i \in S_l} \targetf(i) \leq R \max_{i \in
S_l} f(i)= R f(S_l),$$
where the first inequality follows by definition, and the second
inequality follow from relation~(\ref{RORPMAC:one}).
By definition, for any $i \in S_l$, $f(i) \leq \targetf(S_l)$. Thus,
$f(S_l) = \max_{i \in S_l} f(i) \leq \targetf(S_l)$. These together
imply relation~(\ref{approx-sample-RORPMAC}), as desired.

To finish the proof we show that  $\ell=O(\frac{1}{\epsilon} \left( n \log \log_{1+\eta} (\frac{H}{h}) + \log(1/\delta) \right))$ is sufficient so that  with probability at least $1-\delta$ $f$
approximates the target function $\targetf$ multiplicatively within a factor of $R(1+\eta)^2$ on a $1-\epsilon$ fraction of the distribution.
Let $F_{\eta}$ be the class of unit-demand functions that assign to each individual leaf a power of $(1+\eta)$ in $[h,H]$. Clearly $|F_{\eta}|=(\log_{1+\eta} (\frac{H}{h}))^n$.
 It is easy to see that $\ell=O(\frac{1}{\epsilon} \left( n \log \log_{1+\eta} (\frac{H}{h}) + \log(1/\delta) \right))$ examples are sufficient such that any function in $F_{\eta}$ that approximates the target function on the sample multiplicatively within a factor of $R(1+\eta)$ will with probability at least $1-\delta$
approximate the target function multiplicatively within a factor of $R(1+\eta)$ on a $1-\epsilon$ fraction of the distribution. Since $F_{\eta}$  is a multiplicative $L_{\infty}$ cover for the class of unit-demand functions, we easily get the desired result~\cite{AB99}.
\end{proof}

\vspace{-0.2cm}
\begin{algorithm}
{\bf Input:} A sequence of  training examples $\cS= \left\{(S_1,\targetf(S_1)),
(S_2,\targetf(S_2)), \ldots (S_{\ell},\targetf(S_{\ell})) \right\}$.
\begin{itemize}
\item Set~ $ f(i) = \min_{j: i \in S_j} \targetf(S_j)$~ if $i \in  \cup_{l=1}^m S_l$
~and~ $f(i) = 0$~ if $i \not\in \cup_{l=1}^m S_l$.
\end{itemize}
{\bf Output:} The unit-demand valuation $f$ defined by ~$f(S) = \max_{i \in S} f(i)$~ for any $S \subseteq \{1, \ldots, n\}$.
\caption{~Algorithm for PMAC-learning interesting classes of \XOS\ and \OXS\ valuations. 
}
\label{alg-XOR}
\end{algorithm}
\vspace{-0.2cm}

\section{PMAC-learnability of \OXS\ and Gross Substitutes Valuations}
\label{sec:GS-OXS}
In this section we
study the learnability of subclasses of  submodular valuations, namely \OXS\ and gross substitutes.
We start by focusing on interesting subclasses of \OXS\ functions that arise in practice, namely \OXS\ functions representable with a small number of \XOR trees or leaves\footnote{ We note that the literature on algorithms for secretary problems~\cite{BabaioffDGIT09SecretaryPWaD,secretary} often considers a subclass of the latter class,
in which each item must have the same value 
in any tree.}.
For example, 
a traveler  presented with a collection of
  plane tickets, hotel rooms, and rental cars for a given location
  might value the bundle as the sum
  of his values on the best ticket, the best hotel room, and the best
  rental car. This valuation is \OXS, with one \XOR tree for each travel
 requirement.  The number of \XOR trees, i.e. travel requirements, is small but the number of leaves in each tree may be large. 
 As another example, consider for example a company producing airplanes that must procure many
 different components for assembling an airplane.
 The number of suppliers for each component is small, but the number of
 components may be very large (more than a million in today's
 airplanes).
 The company's value for a set of components of the same type, each
 from a different supplier, is its highest value for any such
 component.
 The company's value for a set of components of different types is the
 sum of the values for each type.
 This valuation  is representable as an \OXSLong,
 with one tree for each  component type. The number of leaves,
 i.e. suppliers, in each \XOR tree is small but there may be many such trees.
In this section, we show good  PMAC-learning  guarantees for classes of functions of these types. Formally:

\vspace{-0.2\baselineskip}
\begin{Theorem}
\textbf{(1)} Let $\cF$ be the family of \OXS\ functions representable with at most $R$  \XOR trees.
For any $\eta$, the family
$\cF$ is properly PMAC-learnable with approximation factor  of $R(1+\eta)$  by using $m=O(\frac{1}{\epsilon} \left( n \log \log_{1+\eta} (\frac{H}{h}) + \log(1/\delta) \right))$ training examples and running time polynomial in $m$, where $h$ and $H$ are the smallest and the largest value our functions can take.
For constant $R$, the class $\cF$ is PAC-learnable by using $O(n^R \log(n/\delta)/\eps)$ training examples and running time $\poly(n,1/\epsilon,1/\delta)$.
\label{stmt:RXORPMAC}
\label{stmt:OXS-R-trees}
\label{stmt:RXORPMACconst}

\textbf{(2)} For any $\epsilon>0$, the class of \OXS\ functions representable with  \XOR\ trees  with at most $R$ leaves is PMAC-learnable with approximation factor $R+\epsilon$  by using $O(\frac{n}{\epsilon} \log \left(\frac{n}{\delta \epsilon} \right))$ training examples and running time $\poly(n,1/\epsilon,1/\delta)$.
\label{stmt:OXS-R-Leaves}
\label{stmt:OXS-R}
\end{Theorem}

\begin{proof}[Proof sketch]
\textbf{(1)} We can show that a function $f$ with  an \OXS\ representation $\cT$ with at most $R$ trees can also be represented as an \XOS\ function with at most $R$ leaves per tree.
Indeed, for each tuple of leaves, one from each tree in $\cT$,
we create an \OR tree with these leaves. The \XOS\ representation of $\targetf$ is the \XOR of all these trees.
Given this the fact that $\cF$ is learnable to a factor of of $R(1+\eta)$ for any $\eta$ follows from Theorem~\ref{stmt:XOS-R-Leaves}.

We now show that when $R$ is constant  the class $\cF$ is PAC-learnable.
First, using a similar argument to the one in Theorem~\ref{stmt:XOS-R-Leaves} we can show that
 Algorithm~\ref{alg-XOR} can be used to PAC-learn any unit-demand valuation  by using
    $
    \ell = O(n \ln(n/\delta)/\eps) $ training
    examples and time $\poly(n,1/\epsilon,1/\delta)$ -- see Lemma~\ref{stmt:XS-1} in Appendix~\ref{app:OXS-R}.
Second, it is easy to see that
an \OXS\ function $\targetf$ representable with at most $R$ trees can also be represented as
a unit-demand with at most $n^R$ leaves, with $R$-tuples as items (see Lemma~\ref{stmt:OXS-R-trees-meta-items} in Appendix~\ref{app:OXS-R}).
These two facts together imply that
for constant $R$, the class $\cF$ is PAC-learnable by using $O(n^R \log(n/\delta)/\eps)$ training examples and running time $\poly(n,1/\epsilon,1/\delta)$.

\textbf{(2)}
We start by showing the following structural result: if $\targetf$ has an \OXS representation with at most $R$ leaves in any \XOR\ tree, then it can be approximated by a linear function within a factor of $R$ on every subset of the ground set. In particular,
the linear function $f$ defined as $f(S) = \sum_{i \in S} \targetf(i)$, for all $S \subseteq \{1 \ldots n\}$ satisfies
\begin{eqnarray}
\targetf(S) \leq f(S) \leq R \cdot \targetf(S)~~~~\mathrm{for~~all~~}~~~~ S \subseteq \{1 \ldots n\}
\label{structuralOXSRleaves}
\end{eqnarray}
By subadditivity, $\targetf(S) \leq R f(S)$, for all S.
Let $\targetf_1, \dots \targetf_k$ be the unit-demand functions that define $\targetf$.
Fix a set $S \subseteq [n]$.
For any item $i \in S$, define $j_i$ to be the index of the $\targetf_j$  under
which item $i$ has highest value: $\targetf(i) = \targetf_{j_i}(\{i\})$.  Then
for the partition $(S_1,...,S_k)$ of $S$ in which item $i$ is mapped to
$S_{j_i}$ for any $i$, we have
$\sum_{i \in S} \targetf(i) \leq R \targetf_1(S_1) + \ldots R \targetf_k(S_k)$.
Therefore:
$$
\textstyle f(S) = \frac{1}{R} \sum\nolimits_{i \in S} \targetf(i) ~\leq~ \max_{(S_{1}, \dots, S_{k}) \text{ partition of } S} ~  ( \targetf_{1}(S_{1}) + \dots + \targetf_{k}(S_{k}) ) ~ = ~ \targetf(S),
$$
where the last equality follows simply from the definition of an \OXS\ function.

Given the structural result~\ref{structuralOXSRleaves}, we can PMAC-learn the class of \OXS\ functions representable with  \XOR\ trees  with at most $R$ leaves y using Algorithm~\ref{alg-lin-sep} with parameters $R$, $\epsilon$ and $p=1$. The correctness by using a reasoning similar to the one in Theorem~\ref{alg-lin-sep-first}.
\end{proof}

\label{sec:GS}
We now consider the class of Gross Substitutes valuations, a superclass of \OXS\ valuations and a subclass of submodular valuations (recall Lemma~\ref{stmt:manyIncl}).
Gross Substitutes are fundamental to allocation problems with per-item prices~\cite{CramtonSS06CombinatorialA,GulS99WalrasianEwGS,book07};
in particular a set of per-item market-clearing prices exists if and (almost) only if all customers have
gross substitutes valuations.
 A valuation is gross substitutes
if raising prices on some items preserves the demand on  other items.
Given prices on items, an agent  with valuation $f$ demands a preferred set, formalized as follows.

\vspace{-0.2\baselineskip}
\begin{Definition}
For price vector $\vec{p} \in\reals^{n}$, the \emph{demand correspondence} $\demandSet{}_{f}(\vec{p})$
of valuation $f$ is the  collection of preferred sets
at prices $\vec{p}$, i.e.
$
 \demandSet{}_{f}(\vec{p}) = \arg\max_{S \subseteq \{1, \ldots, n\}} \{ f(S) - \textstyle\sum\nolimits_{j \in S} p_{j}\}
$.
A valuation $f$ is {\em gross substitutes}  (GS)
if for any price vectors
$\vec{p}' \geq \vec{p}$ (i.e. $p_{i}' \geq p_{i} \forall i \!\in\! [n]$), and any $ A \in \demandSet{}_{f}(\vec{p})$ there exists $A' \in \demandSet{}_{f}(\vec{p}')$ with $ A' \supseteq \{i\in A: p_{i} = p'_{i} \}$.
\label{def:GS}
\end{Definition}
\vspace{-0.5\baselineskip}

That is, the GS property requires that all items $i$
in some preferred set $A$ at the old prices $\vec{p}$ and
for which the old and new prices are equal ($p_i = p_i'$)
are simultaneously contained in some preferred set $A'$ at the new prices $\vec{p}'$.

As mentioned earlier Balcan and Harvey~\cite{BalcanH10LearningSF} proved that it is hard to PMAC-learn  the class
 of submodular functions
with an approximation factor  $o(n^{1/3} / \log n)$.
We show here that their result applies even for the class of gross substitutes.
This is quite surprising since such functions are typically considered easy from an economic optimization point of view.
Specifically:
\begin{Theorem}
\label{nicelowerbound}
 No algorithm can PMAC-learn the class of  gross substitutes with an approximation factor of
    $o({n^{1/3}}/{\log n})$.
  This holds even if $D$ is known and value queries are allowed.
\end{Theorem}
\vspace{-0.3cm}
\begin{proof}
It is known that the class of matroid rank functions cannot be PMAC-learned with an approximation factor of
    $o({n^{1/3}}/{\log n})$, even if $D$ is known and value queries are allowed~\cite{BalcanH10LearningSF}.
One can show that a matroid rank function is a gross substitutes function (see Lemma~\ref{stmt:rankGS} below). Combining these, yields the theorem.
\end{proof}

Our  key tool for proving that any matroid rank function is also GS (Lemma~\ref{stmt:rankGS} below) is a valuation-based characterization of gross substitutes valuations due to~\cite{LienY07OnTheGrossSC}.
\begin{Lemma} 
A matroid rank function is gross substitutes.
\label{stmt:rankGS}
 \end{Lemma}
 \vspace{-\baselineskip}
\begin{proof}
Denote $f$'s marginal value over $S$ by $f^{S}(A) = f(S \cup A) - f(S), \forall  A \!\subseteq\! [n] \! \setminus\! S$.
As shown in~\cite{LienY07OnTheGrossSC} $f$ is GS if and only if 
\begin{align}
f^{S}(ab)+f^{S}(c) \leq \max\{f^{S}(ac)+f^{S}(b), f^{S}(bc)+f^{S}(a)\}\
\text{for all items $a,b,c$ and set $S$}
\label{eq:GSabcS}
\end{align}
i.e. (by taking permutations) 
there is no unique maximizer among $f^{S}(ab)+f^{S}(c), f^{S}(ac)+f^{S}(b), f^{S}(bc)+f^{S}(a)$.

If $f$ is matroid rank function,  then so is $f^S$; in particular, $f^{S}(A) \leq |A|, \forall A \subseteq [n] \!\setminus S$.
We reason by case analysis.

Suppose that $f^S(ab)=2$. Then we have $f^S(a)=f^S(b)=1$. If $f^S(ac)=2$ or $f^S(bc)=2$, then $f^S(c)=1$ and
hence the inequality~\eqref{eq:GSabcS} holds. On the other hand, if $f^S(ac)=f^S(bc)=1$, then we have by the monotonicity and submodularity of $f^{S}$,
$f^S(ab)+f^S(c)\leq f^S(abc)+f^S(c)\leq f^S(ac)+f^S(bc)=2$, and the inequality~\eqref{eq:GSabcS} holds.

If $f^S(ab)\leq 1$ then $f^S(ab)=\max\{f^S(a),f^S(b)\}$. As $f^S(c)\leq f^S(ac)$ and
$f^S(c)\leq f^S(bc)$, Eq.~\eqref{eq:GSabcS} follows.
\end{proof}

We note that Lemma~\ref{stmt:rankGS} was previously proven in~\cite{Murota08SubmodularFMaMiDCA}
in a more involved way via the concept of $\mathrm{M}^\natural$-concavity from discrete convex  analysis.


\section{Learnability everywhere with value queries}
\label{value-queries}
In this section, we consider approximate learning with value queries~\cite{nick09,SF}. This is relevant for settings where instead of passively observing the values of $\targetf$ on sets $S$ drawn from a distribution,
the learner is able to actively query the value $\targetf(S)$ on sets $S$ of its choice and the  goal is to approximate with certainty the target $\targetf$ on all $2^{n}$ sets after querying the values of $\targetf$ on polynomially many sets. Formally:
\begin{Definition}
We say that an
algorithm $\cA$ learns the valuation family $\cF$ everywhere with value queries with an approximation factor of $\alpha \geq 1$
if, for any target function $\targetf \!\in\! \cF$,
after querying the values of $\targetf$ on polynomially (in $n$) many sets,
$\cA$ outputs in time polynomial in $n$ a function $f$ such that
 $f(S) \leq \targetf(S) \leq \alpha f(S), \forall S \subseteq \{1, \ldots, n\}$.
\end{Definition}
Goemans et al.~\cite{nick09} show that
for submodular functions
the learnability factor with value queries is $\tilde{\Theta}(n^{1/2})$.
We show here that their lower bound 
applies to the more restricted  \OXS\ and GS classes (their upper bound automatically applies).
We also show that this lower bound can be circumvented for the interesting subclasses of \OXS\ and \XOS\ that we considered earlier, efficiently achieving a  factor of $R$.
\begin{Theorem} 
\textbf{(1)} The classes of \OXS\ and GS functions are learnable with value queries with an approximation factor of $\tilde{\Theta}(n^{1/2})$.

\textbf{(2)} The following classes are learnable with value queries with an approximation factor of $R$:
\OXS\ with at most $R$ leaves in each tree,
\OXS\ with at most $R$ trees,
\XOS\ with at most $R$ leaves in each tree,
and
\XOS\ with at most $R$ trees.
\label{stmt:multExactHard}
\end{Theorem}
\vspace{-0.5\baselineskip}

\begin{proof}[Proof Sketch]
\textbf{(1)}
We show in Appendix~\ref{app:oneBumpOXS} that the family of valuation functions used in~\cite{nick09} for proving the $\Omega(\frac{n^{1/2}}{\log n})$ lower bound  for learning submodular valuations  with value queries is contained in \OXS. 
The valuations in this family are of the form $g_{23}(S) = \min(|S|, \alpha')$ and $g^R(S) = \min(\beta + |S \cap ((\{1, \ldots, n\})\setminus R)|, |S|, \alpha') $
for $\alpha' = xn^{1/2}/5, \beta = x^2/5$ with $x^2 = \omega(\log n)$ and $R$ a subset of $\{1, \ldots, n\}$ of size $\alpha'$ (chosen uniformly at random).
These valuations are \OXS;
for example, $g_{23}(S)$ can be expressed as a SUM of $\alpha'$ MAX trees, each having as leaves  all items in $[n]$ with weight $1$.

\textbf{(2)}
To establish learnability for these interesting subclasses, we recall that for the first three of them (Theorems~\ref{stmt:XOS-R-Leaves} and~\ref{stmt:OXS-R}) any valuation $\targetf$ in each class was approximated to an $R$ factor by a function $f$ that only depended on the values of $\targetf$ on items.
An analogous result holds for the fourth class, i.e. \XOS\ with at most $R$ trees -- indeed, for such an \XOS\ $\targetf$, we have $\frac{1}{R} \sum_{i \in S} \targetf(\{i\}) \leq \targetf(S) \leq R \frac{1}{R} \sum_{i \in S} \targetf(\{i\}), \forall S \subseteq [n]$.
One can then query these $n$ values and output the corresponding valuation $f$.
\end{proof}

\smallskip

\noindent {\bf Note:}~~We note that the lower bound technique in~\cite{nick09} has been later used in a sequence of papers~\cite{IN09,GKTW09,SF} concerning optimization under submodular cost functions and our result (Lemma~\ref{stmt:GoemansOXS} in particular) implies that all the lower bounds in these papers apply to the smaller classes of \OXS\ functions and \GS functions.


\smallskip

\noindent {\bf Note:}~~We also note that since \XOS\ contains all submodular valuations, the lower bound of Goemans et al.~\cite{nick09} implies that the \XOS\ class is not learnable everywhere with value queries to a $o(\frac{n^{1/2}}{\log n}) =  \tilde{o}(n^{1/2})$ factor.
For the same $\Omega(\frac{n^{1/2}}{\log n})$ lower bound,
our proof technique (and associated family of \XOS\ valuations) for
Theorem~\ref{stmt:XOS-lower-bound} offers a simpler argument than that in~\cite{nick09}.

\section{Learning with prices}
\label{sec:prices}

We now introduce a new paradigm 
that is natural in many applications
where the learner can repeatedly obtain information on
the unknown valuation function of an agent via the agent's
\emph{decisions to purchase or not} rather than
via random samples from this valuation or via queries to it.
In this framework, the learner does not obtain the value of $\targetf$ on each input set $S_{1}, S_{2}, \dots$.
Instead, for each input set $S_{l}$, the learner observes $S_{l}$, quotes a price $p_{l}$ (of its choosing) on $S_{l}$ and obtains one bit of information: whether the agent purchases $S_{l}$ or not, i.e. whether $p_{l} \leq \targetf(S_{l})$ or not.
 The goal remains to approximate the function $\targetf$ well, i.e. within an $\alpha$ multiplicative factor:
   on most sets from $D$ with high confidence for PMAC-learning
   and on all sets with certainty for learning everywhere with value queries.
 The learner's challenge is in choosing prices that allow discovery of the agent valuation.
This framework is a special case of demand queries~\cite{book07}, where prices are: $p_{l}$  on $S_{l}$ and $\infty$ 
 elsewhere.
We call PMAC-learning with prices and VQ-learning with prices the variants of  this framework applied to our two learning models.
Each variant in this framework offers less information to the learner than its respective basic model.

\newcommand{\balpha}{\beta} 

Clearly, all our PMAC-learning lower bounds still hold for PMAC-learning with prices.
More interestingly, our upper bounds still hold as well.  In particular, we provide a reduction from the problem of PMAC-learning with prices to the problem of learning a linear separator,
for functions $\targetf$ such that for some $p>0$, $(\targetf)^p$ can be approximated to a $\balpha$ factor by a linear function.
Such $\targetf$ can be PMAC-learned to a $\balpha^{1/p}$ factor by Algorithm~\ref{alg-lin-sep}.
What we show in Theorem~\ref{stmt:PMACPrices} below is that such $\targetf$ are PMAC-learnable with prices to a factor of $(1+o(1))\balpha^{1/p}$
using only a small increase
in the number of samples over that used for (standard) PMAC learning.
For convenience, we assume in this section that all valuations are integral
and that $H$ is an upper bound on the values of $\targetf$, i.e.
$\targetf(S) \leq H, \forall S \subseteq [n]$.

\begin{Theorem}
Consider a family $\cF$ of valuations such that the $p$-th power of any $\targetf \in \cF$ can be approximated to a $\balpha$ factor
by a linear function: i.e., for some $w$ we have  $w\transpose \ind(S) \leq (\targetf(S))^{p} \leq \balpha w\transpose \ind(S)$ for all $S \subseteq [n]$, where $\balpha \geq 1, p > 0$.
Then for any $0<\eta\leq 1$, the family $\cF$ is PMAC-learnable with prices to a $(1+\eta)\balpha^{1/p}$ factor
using $O(\frac{n\log H}{\eta\eps} \ln(\frac{n\log H}{\eta\eps\delta}) )$
samples and time $\poly(n, \frac{1}{\epsilon}, \frac{1}{\delta}, \frac{1}{\eta})$.
\label{stmt:PMACPrices}
\end{Theorem}
\vspace{-0.5\baselineskip}
\noindent
\begin{proof}

\newcommand{\alphaPFactor}{\balpha}

As in Algorithm~\ref{alg-lin-sep}, the idea is to use a reduction to learning a linear separator, but
where now the examples use prices (that the algorithm can choose) instead of function values (that the algorithm can no longer observe).
For each input set $S_{l}$, the purchase decision amounts to a comparison between the chosen price $q_{l}$ and $\targetf(S_{l})$.
Using the result of this comparison
we will construct examples, based on the prices $q_{l}$, that are always consistent with a linear separator obtained from $w\transpose \ind(S_{l})$, the linear function that approximates $(\targetf)^p$.
We will sample enough sets $S_{l}$ and assign prices $q_{l}$ to them in such a way that for sufficiently many $l$, the price $q_{l}$ is close to $\targetf(S_{l})$.
We then find a linear separator that has small error on the distribution induced by the
price-based  examples and show this yield a hypothesis $f(S)$ whose
error is not much higher on the original distribution with respect to the values of the (unknown) target function.

Specifically, we take $m = O(\frac{n\log H}{\eta\eps} \ln(\frac{n\log
H}{\eta\eps\delta}) )$ samples, and for convenience define $N =
\lfloor \log_{1+\eta/3} H \rfloor$. We assign to each input bundle
$S_{l}$ a price $q_{l}$ drawn uniformly at random from $\{(1 +
\eta/3)^i\}$ for $i=0,1,2, \ldots, N+1$, and present bundle $S_l$ to
the agent at price $q_l$.  The key point is that for bundles $S_l$
such that $\targetf(S_{l}) \geq 1$ this ensures at least a
$\frac{1}{N+2}$ probability that $\targetf(S_{l})(1+\eta/3)^{-1} <
q_{l} \leq \targetf(S_{l})$ and at least a $\frac{1}{N+2}$ probability
that $\targetf(S_{l}) < q_{l} \leq \targetf(S_{l})(1+\eta/3)$ (the
case of $\targetf(S_{l})=0$ will be noticed when the agent does not
purchase at price $q_l=1$ and is handled as in the proof of Theorem
\ref{alg-lin-sep-first}).

We construct new examples based on these prices and purchase decisions as follows. If $\targetf(S_{l}) < q_{l}$ (i.e. the agent does not buy) then we let $(x_{l}, y_{l})=((\ind(S_{l}), \alphaPFactor q_{l}^{p}), -1)$.
If $\targetf(S_{l}) \geq q_{l}$ (i.e. the agent buys) then we let $(x_{l}, y_{l}) = ( (\ind(S_{l}), q_{l}^{p}), +1)$.
Note that by our given assumption, the examples constructed
are always linearly separable.  In particular
the label $y_{l}$ matches $\sgn( (\alphaPFactor w,-1)\transpose x_{l} )$ in each case:
 $\alphaPFactor w\transpose\! \ind(S_{l}) \leq \alphaPFactor  (\targetf(S_{l}))^{p} < \alphaPFactor  q_{l}^{p}$
and
 $ q_{l}^{p} \leq (\targetf(S_{l}))^{p} \leq \alphaPFactor  w\transpose\! \ind(S_{l})$ respectively.
Let $D^{n+1}_{buy}$ denote the induced distribution on $\bR^{n+1}$.
We now find a linear separator $(\hat{w},-z) \!\in\! \bR^{n+1}$, where $\hat{w} \!\in\! \bR^n$ and $z \!\in\! \bR_+$, that is consistent
with $(x_{l},y_{l}), \forall l$.
We construct an intermediary hypothesis
$f'(S) = \frac{1}{\alphaPFactor z}\hat{w} \transpose \ind(S)$ based on the learned linear separator. The hypothesis output will be $f(S) = \frac{1}{1+\eta/3} (f'(S))^{1/p}$.

By standard VC-dimension sample-complexity bounds, our sample size $m$ is sufficient that the linear separator $(\hat{w},-z)$ has error on $D^{n+1}_{buy}$ at most $\frac{\eps}{N+2}$ with probability at least $1-\delta$.  We now show that this implies that with probability at least $1-\delta$, hypothesis $f(S)$ approximates $\targetf(S)$ to a factor $(1+\eta/3)^2\balpha^{1/p} \leq (1+\eta)\balpha^{1/p}$ over $D$, on all but at most an $\eps$ probability mass, as desired.

Specifically, consider some bundle $S$ for which $f(S)$ does {\em not} approximate $\targetf(S)$ to a factor $(1+\eta/3)^2\balpha^{1/p}$ and for which $\targetf(S) \geq 1$ (recall that zeroes are handled separately).  We just need to show that for such bundles $S$, there is at least a $\frac{1}{N+2}$ probability (over the draw of price $q$) that $(\hat{w},-z)$ makes a mistake on the resulting example from $D^{n+1}_{buy}$.  There are two cases to consider:
\begin{enumerate}
\item It could be that $f$ is a bad approximation because $f(S) > \targetf(S)$.  This implies that $f'(S) > [(1 + \eta/3)\targetf(S)]^p$ or equivalently that $\hat{w} \transpose \ind(S) > \balpha z [(1 + \eta/3)\targetf(S)]^p$.  In this case we use the fact that there is a $\frac{1}{N+2}$ chance that $\targetf(S) < q \leq \targetf(S)(1+\eta/3)$.  If this occurs, then the agent doesn't buy (yielding $x = (\ind(S),\balpha q^p), y=-1$) and yet $\hat{w} \transpose \ind(S) > \balpha z q^p$.  Thus the separator mistakenly predicts positive.
\item Alternatively it could be that $(1+\eta/3)^2\balpha^{1/p}f(S) < \targetf(S)$.  This implies that $(1+\eta/3)\balpha^{1/p}f'(S)^{1/p} < \targetf(S)$ or equivalently that $\hat{w}\transpose \ind(S) < z(\frac{\targetf(S)}{1+\eta/3})^p$.  In this case, we use the fact that there is a $\frac{1}{N+2}$ chance that $\frac{\targetf(S)}{1+\eta/3} < q \leq \targetf(S)$.  If this occurs, then the agent does buy (yielding $x = (\ind(S),q^p), y=+1$) and yet $\hat{w} \transpose \ind(S) < zq^p$.  Thus the separator mistakenly predicts negative.
\end{enumerate}
Thus, the error rate under $D^{n+1}_{buy}$ is at least a $\frac{1}{N+2}$ fraction of the error rate under $D$, and so a low error under  $D^{n+1}_{buy}$ implies a low error under $D$ as desired.
\end{proof}

\noindent {\bf Note:}~~We note that if there is an underlying desired pricing algorithm $\cA$,
for each input set $S_{l}$ we can take the price of $S_{l}$ to be $\cA(S_{l})$ with probability $1-\tilde{\eps}$ and a uniformly at random price in $\{1,2,4 \dots, H/2, H\}$ as in the previous result with probability $\tilde{\eps}$. The sample complexity of learning only goes up by a factor of at most $\frac{\log H \log{ \log H}}{\tilde{\eps}}$.

\smallskip

We can also recover our upper bounds on learnability everywhere with value queries (the corresponding lower bounds clearly hold).
By sequentially setting prices $1,2,4 \dots, H/2,H$ on each item we can learn
 $\targetf$'s values on items within a factor of $2$. Our structural results proving the approximability of $\targetf$ from interesting classes with a function that only depends on $\targetf(\{1\}),\dots,\targetf(\{n\})$ then yield the VQ-learnability with prices of these classes.

\begin{Theorem}
The following classes are VQ-learnable with prices to within an $2R$ factor:
\OXS\ with at most $R$ trees, \OXS\ with at most $R$ leaves in each tree,
\XOS\ with at most $R$ trees, and \XOS\ with at most $R$ leaves in each tree.
\label{stmt:VQPrices}
\end{Theorem}

\section{Conclusions}

In this paper we study the approximate learnability of valuations
 commonly used throughout economics and game theory
for the quantitative encoding of agent preferences. We provide upper and lower bounds regarding the learnability
of important subclasses of valuation functions that express no-complementarities.
%
Our main results concern their approximate learnability in the distributional learning (PAC-style) setting. We provide nearly tight lower and upper bounds of $\tilde{\Theta}(n^{1/2})$ on the approximation factor for learning \XOS\ and subadditive valuations, both widely studied superclasses of submodular valuations. Interestingly, we show that the  $\tilde{\Omega}(n^{1/2})$ lower bound can be circumvented for \XOS\ functions of polynomial complexity; we provide an algorithm for learning the class of \XOS\ valuations with a representation of polynomial size achieving an $O(n^{\eps})$ approximation factor in time $O(n^{1/\eps})$ for any $\eps > 0$. This highlights the importance of considering the complexity of the target function for polynomial time learning.  We also provide new learning results for interesting subclasses of submodular functions.
Our upper bounds for distributional learning leverage novel structural results for all these valuation classes. We
show that many of these results provide new learnability results in the Goemans et al. model~\cite{nick09} of approximate learning everywhere via value queries.

We also introduce a new model that is more realistic in economic settings, in which the learner
can set prices and observe purchase decisions at these prices rather than observing the valuation function directly.
In this model, most of our upper bounds continue to hold
despite the fact that the learner receives less information (both for learning in the distributional setting and with value queries),
while our lower bounds naturally extend.

\smallskip
\noindent\textbf{Acknowledgments}.
We thank Avrim Blum, Nick Harvey, and David Parkes for useful discussions.
This work was supported in part by NSF grants CCF-0953192 and CCF-1101215, AFOSR grant FA9550-09-1-0538, and a Microsoft Research Faculty Fellowship. 

\bibliographystyle{abbrv}
\bibliography{GS}

\appendix

\section{Additional Details for the Proof of Theorem~\ref{alg-lin-sep-first}}
\label{appendix-lin-sep}

\renewcommand{\targetfpower}[1]{(\targetf(#1))^{2}}

For PMAC-learning \XOS\ valuations to an $\sqrt{n+\eps}$ factor,
we apply Algorithm~\ref{alg-lin-sep} with parameters $R=n$, $\epsilon$, and $p=2$. 
The proof of correctness of  Algorithm~\ref{alg-lin-sep} follows by using the structural result in Claim~\ref{Satoru} and a technique of \cite{BalcanH10LearningSF}. We provide here the full details of this proof.

Because of the multiplicative error allowed by the PMAC-learning model, we
separately analyze the subset of the instance space where $\targetf$ is zero and the
subset of the instance space where $\targetf$ is non-zero. For convenience, we
define:
$$
\cP = \setst{ S }{ \targetf(S) \neq 0 } \qquad\text{and}\qquad \cZ = \setst{ S }{
\targetf(S) = 0 }.
$$
The main idea of our algorithm is to reduce our learning problem to the standard problem
of learning a binary classifier (in fact, a linear separator) from i.i.d.\ samples in the
passive, supervised learning setting~\cite{KV:book94,Vapnik:book98} with a slight twist
in order to handle the points in $\cZ$. The problem of learning a linear separator in the
passive supervised learning setting is one where the instance space is $\bR^m$, the
samples come from some fixed and unknown distribution $D'$ on $\bR^m$, and there is a
fixed but unknown target function $c^* : \bR^m \rightarrow \set{-1,+1}$, $c^*(x) = \sgn(u
\transpose x)$.
 The examples induced by $D'$ and $c^*$ are called \emph{linearly
separable} since there exists a vector $u$ such that $c^*(x) = \sgn(u \transpose x)$.
The linear separator learning problem we reduce to is defined as follows. The instance
space is $\bR^m$ where $m=n+1$ and the distribution $D'$ is defined by the following
procedure for generating a sample from it. Repeatedly draw a sample $S \subseteq [n]$
from the distribution $D$ until $\targetf(S) \neq 0$. Next, flip a fair coin for each. The sample
from $D'$ is
\begin{align*}
(\ind(S), \targetfpower{S}) &\qquad\text{(if the coin is heads)} \\
(\ind(S), (n+\eps) \cdot \targetfpower{S}) &\qquad\text{(if the coin is tails).}
\end{align*}
The function $c^*$ defining the labels is as follows:
samples for which the coin was heads are labeled $+1$, and the others are labeled $-1$.
We claim that the distribution over labeled examples induced by $D'$ and $c^*$ is
linearly separable in $\bR^{n+1}$.
To prove this we use the assumption that for the linear function $f(S) = \hat{w}\transpose \ind(S)$ with $w \in  \mathbb{R}^{n}$, we have $(\targetf(S))^{2} \leq \hat{f}(S) \leq n (\targetf(S))^{2}$ for all $S \subseteq [n]$.
Let $u = ( (n+\eps/2)\cdot w, -1) \in \bR^m$.
For any point $x$ in the support of $D'$ we have
\begin{align*}
x =(\ind(S), \targetfpower{S})
&\qquad\implies\qquad
u \transpose x = (n+\eps/2) \cdot \hat{f}(S) - \targetfpower{S} > 0
\\
x=(\ind(S), (n+\eps)\cdot\targetfpower{S})
&\qquad\implies\qquad
u \transpose x = (n+\eps/2) \cdot \hat{f}(S) - (n+\eps)\cdot\targetfpower{S} < 0.
\end{align*}
This proves the claim. Moreover, this linear function also satisfies $\hat{f}(S)=0$ for
every $S \in \cZ$. In particular, $\hat{f}(S)=0$ for all $S \in \cSz$ and moreover,
$$  \hat{f}(\set{j})~=~w_j~=~0 \qquad\text{ for every } j ~\in ~\cU_{D}\qquad\text{ where}
\qquad  \!\!\!\!\cU_D ~=~ \Union_{\substack{S_i \in \cZ }}{S_i}.
$$

Our algorithm is as follows. It first partitions the training set $\cS=
\left\{(S_1,\targetf(S_1)), \ldots, (S_\ell,\targetf(S_\ell))\right\}$ into two sets
$\cSz$ and $\cSnz$, where $\cSz$ is the subsequence of $\cS$ with $\targetf(S_i) = 0$,
and $\cSnz=\cS \setminus \cSz$. For convenience, let us denote the sequence  $\cSnz$ as
 $$ \cSnz ~=~ \big( (A_1,\targetf(A_1)), \ldots, (A_a, \targetf(A_a)) \big).$$
Note that $a$ is a random variable and we can think of the sets the $A_i$ as drawn
independently from $D$, conditioned on belonging to $\cP$. Let $$ \cU_{0} ~=~
\Union_{\substack{i \leq \ell \\ \targetf(S_i)=0}} \!\!\!\! S_i \qquad\text{ and }\qquad
\cL_0 ~=~ \setst{ S }{ S \subseteq \cU_0 }.
$$
 Using $\cSnz$, the algorithm then constructs a sequence
$ \cSnz' ~=~ \big( (x_1,y_1),\ldots,(x_a,y_a) \big) $ of training examples for the binary
classification problem. For each $1 \leq i \leq a$, let $y_i$ be $+1$ or $-1$, each with
probability $1/2$. If $y_i=+1$ set $x_i = (\ind(A_i), \targetfpower{A_i})$; otherwise set $x_i
= (\ind(A_i), (n+\eps) \cdot \targetfpower{A_i})$. The last step of our algorithm is to  solve a
linear program in order to find a linear separator $u = (\hat{w},-z)$ where $\hat{w} \in \bR^n$, $z
\in \bR$ consistent with the labeled examples $(x_i,y_i)$, $i=1 \leq i \leq a$, with the
additional constraints that $w_j=0$ for $j \in \cU_{0}$. The output hypothesis is $f(S) =
(\frac{1}{(n+\eps) z} \hat{w} \transpose \ind(S))^{1/2}$.

To prove correctness, note first that the linear program is feasible; this follows from
our earlier discussion using
 the facts (1) $\cSnz'$ is a set of labeled examples drawn
from $D'$ and labeled by $c^*$ and  (2) $\cU_{0} \subseteq \cU_D$. It remains to show
that $f$ approximates the target on most of the points. Let $\cY$ denote the set of
points $S \in \cP$ such that both of the points $(\ind(S), \targetfpower{S})$ and $(\ind(S),
(n+\eps)\cdot\targetfpower{S})$ are correctly labeled by $\sgn(u \transpose x)$, the linear
separator found by our algorithm. It is easy to show that the function $f(S)
= (\frac{1}{(n+\eps) z} \hat{w} \transpose \ind(S))^{1/2} $ approximates $\targetf$ to within a factor
$n+\eps$ on all the points in the set  $\cY$. To see this notice that for any point $S \in
\cY$, we have
\begin{gather*}
\hat{w} \transpose \ind(S) - z \targetfpower{S} > 0 \qquad\text{and}\qquad \hat{w} \transpose \ind(S) - z
(n+\eps) \targetfpower{S} < 0
\\
\implies\quad \frac{1}{(n+\eps) z} \hat{w} \transpose \ind(S) ~<~ \targetfpower{S} ~<~ (n+\eps)
\frac{1}{(n+\eps) z} \hat{w} \transpose \ind(S).
\end{gather*}
So, for any point in $S \in \cY$, the function $f(S)^{2} = \frac{1}{(n+\eps) z} \hat{w} \transpose
\ind(S) $ approximates $\targetfpower{\cdot}$ to within a factor $n+\eps$.
Moreover,  by design the function $f$ correctly labels as $0$ all the examples in
$\cL_0$. To finish the proof, we now note two important facts: for our choice of $\ell
= \frac{16n}{\epsilon} \log \left(\frac{n}{\delta \epsilon} \right)$, with high
probability both $\cP \setminus \cY$  and $\cZ \setminus \cL_0$ have small measure.
\begin{claim}
With probability at least $1-\delta$, the set $\cZ \setminus \cL_0$ has measure at most $\epsilon$.
\end{claim}
\begin{proof}
Let $\cL_k ~=~ \setst{ S }{ S \subseteq \cU_k }.$
Suppose that, for some $k$, the set $\cZ \setminus \cL_k$ has measure at least $\epsilon$.
Define $k' = k+\log(n/\delta)/\epsilon$.
Then amongst the subsequent examples $S_{k+1},\ldots,S_{k'}$,
the probability that none of them lie in $\cZ \setminus \cL_k$ is
at most $(1-\epsilon)^{\log(n/\delta)/\epsilon} \leq \delta/n$.
On the other hand, if one of them does lie in $\cZ \setminus \cL_k$,
then $\card{\cU_{k'}} > \card{\cU_k}$.
But $\card{\cU_k} \leq n$ for all $k$, so this can happen at most $n$ times.
Since $\ell \geq n \log(n/\delta) / \epsilon$, with probability at least $\delta$
the set $\cZ \setminus \cL_\ell$ has measure at most $\epsilon$.
\end{proof}
We now prove:
\begin{claim}
If $\ell = \frac{16n}{\epsilon} \log \left(\frac{n}{\delta \epsilon} \right)$, then
 with
probability at least $1-2\delta$, the set $\cP \setminus \cY$ has measure at most
$2\epsilon$ under $D$.
\label{stmt:PMinusY}
\end{claim}
\label{sec:PMinusY}
\begin{proof}[Proof of Claim~\ref{stmt:PMinusY}]
Let $q = 1-p = \probover{S \sim D}{S \in \cP}$. If $q < \epsilon$ then the claim is
immediate, since $\cP$ has measure at most $\epsilon$. So assume that $q \geq \epsilon$.
Let $\mu = \expect{a} = q \ell$. By assumption $ \mu > 16n \log(n/\delta\epsilon)
\frac{q}{\epsilon}$. Then Chernoff bounds give that
\begin{align*}
\prob{ a < 8n \log (n/\delta\epsilon) \frac{q}{\epsilon} }
   < \exp(- n \log(n/\delta)q/\epsilon) < \delta.
\end{align*}
So with probability at least $1-\delta$, we have
   $a \geq 8n \log (qn/\delta\epsilon) \frac{q}{\epsilon}$.
By a standard sample complexity argument~\cite{Vapnik:book98}
with probability at least $1-\delta$, any linear separator consistent with $\cS'$ will
be inconsistent with the labels on a set of measure at most $\epsilon/q$ under $D'$.
In particular, this property holds for the linear separator $c$ computed by the linear program.
So for any set $S$,
the conditional probability that either
$(\ind(S), \targetfpower{S})$ or $(\ind(S), (n+\eps)\cdot\targetfpower{S})$ is incorrectly labeled,
given that $S \in \cP$, is at most $2 \epsilon / q$.
Thus
$$
\prob{ S \in \cP \And S \not \in \cY }
~=~ \prob{ S \in \cP } \cdot \probg{ S \not \in \cY }{ S \in \cP }
~\leq~ q \cdot (2\epsilon/q),
$$
as required.
\end{proof}
In summary,
our algorithm outputs a hypothesis $f$ approximating
$\targetf$ to within a factor $(n+\eps)^{1/2}$ on  $\cY \cup \cL_\ell$.
The complement of this set is $(\cZ \setminus \cL_0) \cup (\cP \setminus \cY)$,
which has measure at most $3 \epsilon$, with probability at least $1-3\delta$.


\section{Additional Results for Theorem~\ref{stmt:OXS-R}}
\label{app:OXS-R}

We prove that Algorithm~\ref{alg-XOR} can be used to PAC-learn (i.e. PMAC-learn with $\alpha=1$) any unit-demand valuation.

\begin{Lemma}
\label{stmt:XS-1}
The class of unit-demand valuations is properly PAC-learnable by using
    $
    \ell = O(n \ln(n/\delta)/\eps) $ training
    examples and  time $\poly(n,1/\epsilon,1/\delta)$.
\end{Lemma}

\begin{proof} We first show how to solve the consistency problem in
polynomial time: given a sample $(S_1, \targetf(S_1)), \ldots, (S_m,
\targetf(S_m))$ we show how to construct in polynomial time a
unit-demand function $f$ that is consistent with the sample, i.e.,
$f(S_l)=\targetf(S_l)$, for $l \in \{1,2, \ldots, m\}$.
In particular, using the reasoning in  Theorem~\ref{stmt:XOS-R-Leaves}
for $R=1$, we show that the unit-demand hypothesis $f$ output by
Algorithm~\ref{alg-XOR} is consistent with the samples.
We have
$$f(S_l) = \max_{i \in S_l} f(i) = \max_{i \in S_l} \min_{j: i \in S_j} \targetf(S_j) \leq \targetf(S_l).$$
Also note that for any $i \in S_l$ we have $\targetf(i) \leq \targetf(S_l)$, for $l \in \{1,2, \ldots, m\}$.  So
$f(i) \geq \targetf(i) ~~\mathrm{~for~any~} i \in S_1 \cup \ldots \cup S_m.$.
Therefore for any $l \in \{1,2, \ldots, m\}$ we have :
$$\targetf(S_l) = \max_{i \in S_l} \targetf(i) \leq   \max_{i \in S_l} f(i)=  f(S_l).$$
Thus $\targetf(S_l)=f(S_l)$ for $l \in \{1,2, \ldots, m\}$.

We now claim that $m = O(n \ln(n/\delta)/\eps)$ training examples are
sufficient so that with probability at least $1-\delta$, the
hypothesis $f$ produced has error at most $\epsilon$.  In particular,
notice that Algorithm~\ref{alg-XOR} guarantees that $f(i) \in
\{\targetf(1), \ldots, \targetf(n)\}$ for all $i$.  This means that
for any given target function $\targetf$, there are at most $n^n$
different possible hypotheses $f$ that Algorithm~\ref{alg-XOR} could
generate.  By the union bound, the probability that the algorithm
outputs one of error greater than $\eps$ is at most $n^n(1-\eps)^m$
which is at most $\delta$ for our given choice of $m$.
\end{proof}

\begin{Lemma}
If $\targetf$ is \OXS\ with at most $R$ trees, then it is
also unit-demand with at most $n^R$ leaves (with $R$-tuples as items).
For constant $R$, the family $\cF$ of
\OXS\ functions with at most $R$ trees
is PAC-learnable using $O(R n^R \log(n/\delta)/\eps)$ training examples and time $\poly(n,1/\epsilon,1/\delta)$.
\label{stmt:OXS-R-trees-meta-items}
\end{Lemma}
\begin{proof}
We start by noting that since $\targetf$ is an \OXS\ function representable with at most $R$ \XOR{} trees, then $\targetf$ is uniquely determined by its values on sets of size up to $R$. Formally,
\begin{align}
\targetf(S)  = \mathop{\max_{S^R \subseteq S}}_{|S^R| \leq R} \targetf(S^R), \forall S \subseteq \{ 1, \dots, n\}
\label{eq:SRf}
\end{align}

We construct a unit-demand $f'$, closely related to $f$, on meta-items
corresponding to each of the $O(n^R)$ sets of at most $R$ items.
In particular, we define one meta-item $i_{S^R}$ to represent each set $S^R \subseteq \{ 1, \dots, n\}$ of size at most $R$ and let
\begin{align*}
f'(i_{S^R}) = \targetf(S^R).
\end{align*}
We define $f'$ as unit-demand over meta-items; i.e. beyond singleton
sets, we have
\begin{align}
f'(\{i_{S_1^R}, \dots, i_{S_L^R}\}) = \max_{l=1, \dots, L} f'(i_{S_l^R}).
\label{eq:SRfp}
\end{align}
By equations \eqref{eq:SRf} and \eqref{eq:SRfp}, for all $S$ we have
$\targetf(S) = f'(I_S)$ where $I_S = \{i_{S^R} : S^R \subseteq S\}$.

Since we can perform the mapping from sets $S$ to their corresponding
sets $I_S$ over meta-items in time $O(n^R)$, this implies that
to PAC-learn $\targetf$, we can simply PAC-learn $f'$ over
the $O(n^R)$ meta-items using Algorithm \ref{alg-XOR}.
Lemma~\ref{stmt:XS-1} guarantees that this will PAC-learn using
$O(n^R \log(n^R/\delta)/\eps)$ training examples and running time
$\poly(n,1/\epsilon,1/\delta)$ for constant $R$.
\end{proof}

\section{Additional Result for Theorem~\ref{stmt:multExactHard}}
\label{app:oneBumpOXS}

\cite{nick09} proved
that a certain matroid rank function $f_{R, \alpha', \beta}(\cdot)$,
defined below,  is hard to learn everywhere with value queries to an approximation factor of $o(\sqrt{\frac{n}{\ln n}})$.
 We show that the rank function $f_{R, \alpha', \beta}(\cdot)$ is in
\OXS\ (all leaves in all \OXS\ trees will have value $1$).
$f_{R, \alpha', \beta}(\cdot): 2^{\{1, \dots, n\}} \to \reals$ is defined as follows.
Let subset $R \subseteq \{1,\dots,n\}$ and $\bar{R} = (\{1, \dots, n\})\!\setminus\!R$ its complement.
Also fix integers $\alpha', \beta \in \bN$.
Then
\begin{align}
f_{R, \alpha', \beta}(S) = \min(\beta + |S \cap \bar{R}|, |S|, \alpha'), \ \forall S \subseteq \{1,\dots,n\}
\label{eq:Goemans}
\end{align}

 As a warm-up, we show that a simpler function than $f_{R, \alpha',
\beta}(\cdot)$ is in \OXS.  This simpler function essentially corresponds to
$\beta = 0$ and will be used in the
case analysis for establishing that $f_{R, \alpha', \beta}(\cdot)$ is
in \OXS.
\begin{Lemma}
Let $R' \subseteq \{1,\dots,n\}$ and $c \in \bN$. Then the function $f(\cdot): 2^{\{1,\dots,n\}} \to \reals$ defined as
\begin{align}
f_{R',c}(S) = \min(c, |S \cap R'|), \forall S \subseteq \{1,\dots,n\}
\label{eq:addBudget}
\end{align}
is in \OXS.
\label{stmt:addBudgetOXS}
\end{Lemma}
\begin{proof}
For ease of notation let $f(\cdot) = f_{R',c}(\cdot)$.
We assume $c < |R'|$; otherwise,
$f(S) = |S \cap R'|, \forall S \subseteq \{1,\dots,n\}$, which is a
linear function ($f(S) = \sum_{x \in S} f(x)$ where $f(x)=1$ if $x \in
R'$ and $f(x)=0$ otherwise) and any linear function belongs to the
class \OXS~\cite{LehmannLN02CombinatorialAwDMU}.
Assuming $c < |R'|$, we construct an \OXS\ tree $T$ with $c$ \XOR{}
trees, each with one leaf for every element in $R'$. All leaves have
value 1.  We refer the reader to Fig.~\ref{fig:oneBumpSimpleOXS}.

\begin{figure}%
\centering
\subfigure[OXS representation for the function in Eq.~\eqref{eq:addBudget}.]{
\label{fig:oneBumpSimpleOXS}
\includegraphics[scale=.6]{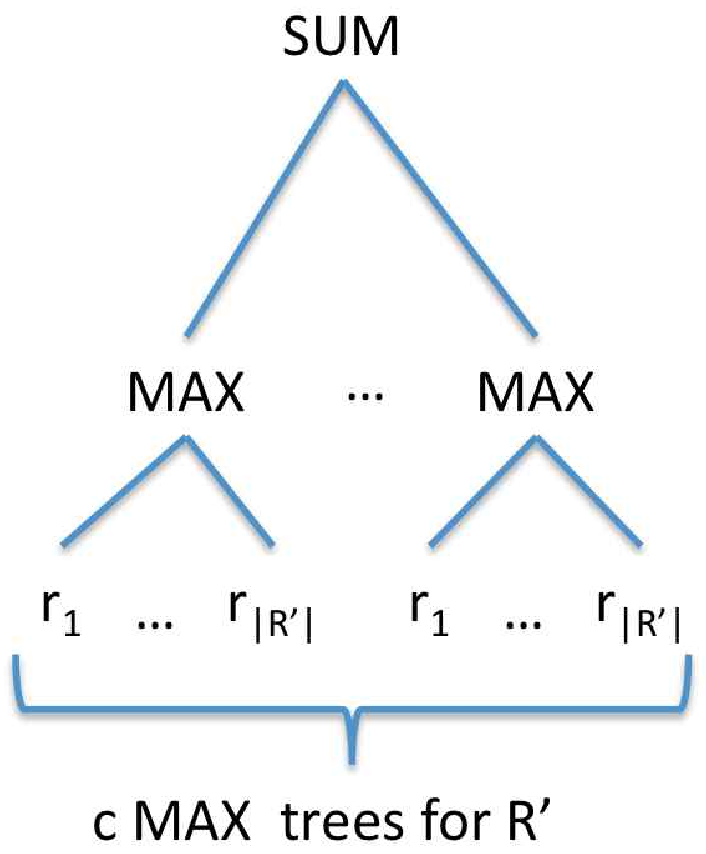}
}
\qquad
\subfigure[OXS representation for the function in Eq.~\eqref{eq:Goemans} when $n > \alpha' > \beta$.]{
\label{fig:oneBumpGoemansOXS}
\includegraphics[scale=.6]{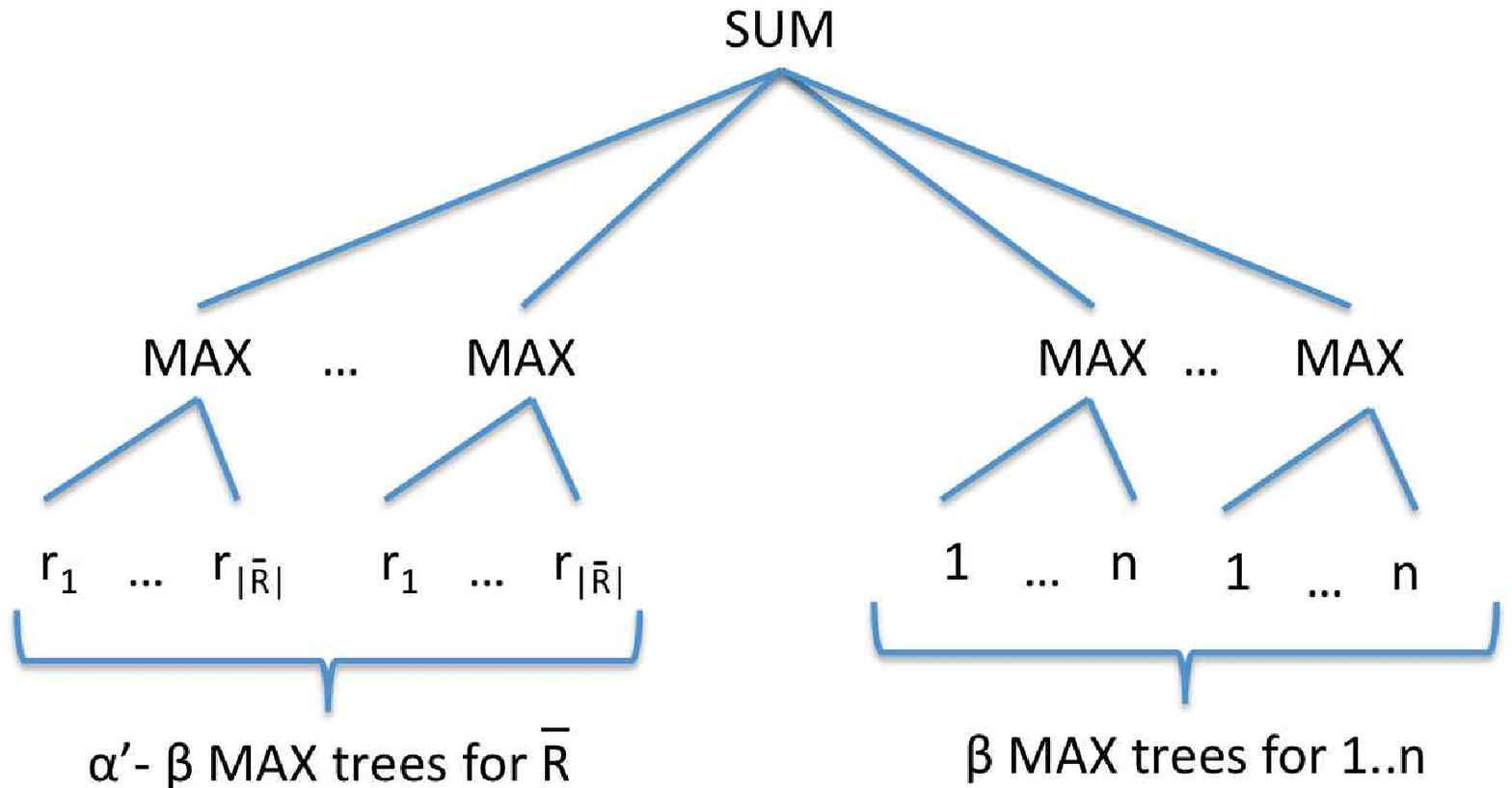}
 }
 \caption{OXS representations.
 All leaves have value $1$.}
\end{figure}%

Then $T(S) = f(S), \forall S \subseteq \{1,\dots,n\}$ since $f(S)$
represents the smaller of the number of elements in $S \cap R'$ (that
can each be taken from a different \XOR{} tree in $T$) and $c$. Note that
$T(S) \leq c, \forall S \subseteq \{1,\dots,n\}$.
\end{proof}
\begin{Lemma}
The matroid rank function $f_{R, \alpha', \beta}(\cdot) $ is in the class \OXS.
\label{stmt:GoemansOXS}
\end{Lemma}
\begin{proof}
For ease of notation let $f(\cdot) = f_{R, \alpha', \beta}(\cdot)$.
If $n \leq \alpha'$ then\footnote{We note that $n >\alpha'$ in~\cite{nick09}.  We consider this case for completeness.} $|S| \leq \alpha', \forall S \subseteq \{1,\dots,n\}$ and
\begin{align}
f(S) = \min(\beta + |S \cap \bar{R}|, |S|) =  |S \cap \bar{R}| + \min( \beta, |S \cap R| )
\label{eq:simplerRIfNLeqAlpha}
\end{align}
From  the proof of Lemma~\ref{stmt:addBudgetOXS} we get that
the function
$f'(S) = \min( \beta, |S \cap R| )$ has an \OXS\ tree $T'$ (i.e. $T'(S) = f'(S), \forall S \subseteq \{1,\dots,n\}$)
with $\beta$ \XOR{} trees each with leaves only in $R$. We can create a new tree $T$
by  adding $|\bar{R}|$ \XOR{} trees to $T'$, each with one leaf
for every element in $\bar{R}$, and we get $T(S) = f(S), \forall S \subseteq \{1,\dots,n\}$.
The additional $|\bar{R}|$ \XOR{} trees encode the $ |S \cap \bar{R}| $ term in Eq.~\eqref{eq:simplerRIfNLeqAlpha}.
If $\alpha' \leq \beta$ then $f(S) = \min(\alpha', |S|)$; the claim follows by Lemma~\ref{stmt:addBudgetOXS} for $c = \alpha', R' = \{1,\dots,n\}$.
We can thus assume that $n > \alpha' > \beta$.
We prove that the \OXS\ tree $T$, containing the two types of \XOR{} trees below, represents $f$, i.e. $T(S)= f(S), \forall S$. We refer the reader to Fig.~\ref{fig:oneBumpGoemansOXS}.
\begin{itemize}
\item $\alpha' - \beta$ \XOR{} trees $T_{1}\dots T_{\alpha' - \beta}$, each having as leaves all the elements in $\bar{R}$ with value $1$.
\item $\beta$ \XOR{} trees $T_{\alpha' - \beta+1} \dots T_{\alpha'}$, each having as leaves all the elements (in $\{1,\dots,n\}$) with value $1$.
\end{itemize}
We note that $T(S) \leq \min(|S|, \alpha')$ as no set $S$ can use more than $|S|$ leaves and $T$ has exactly $\alpha'$ trees.
We distinguish the following cases
\begin{itemize}
\item $|S| \leq \beta$ implying $f(S) = |S|$ and $T(S) = |S|$ as $|S|$ leaves can be taken each from $|S|$ trees in $T_{\alpha' - \beta+1} \dots T_{\alpha'}$.
\item  $f(S) = \alpha' \leq \min(\beta+|S \cap \bar{R}|,|S|)$.  We
claim $T(S) \geq \alpha'$. There must exist $\alpha' - \beta$ elements
in $|S \cap \bar{R}|$, that we can select one from each tree
$T_{1}\dots T_{\alpha' - \beta}$. Also $|S| \geq \alpha'$ and we can
take the remaining $\beta$ elements from $T_{\alpha' - \beta+1} \dots
T_{\alpha'}$.
\item $f(S) = |S| \leq \min(\beta + |S \cap \bar{R}|, \alpha')$. 
This implies $|S \cap R| \leq \beta$ and $|S|\leq \alpha'$.
We claim $T(S) \geq |S|$: we can take all needed elements in $S \cap \bar{R}$ from $T_{1}\dots T_{\alpha' - \beta}$ (and from $T_{\alpha' - \beta + 1} \dots T_{\alpha'}$ if $|S\cap \bar{R}| > \alpha' - \beta$) and elements in $S \cap R$ from $T_{\alpha' - \beta + 1} \dots T_{\alpha'}$.

\item $f(S) = \beta + |S \cap \bar{R}| \leq \min(|S|, \alpha')$.
We claim $T(S) \geq \beta + |S \cap \bar{R}|$: we can take
$ \beta \leq |S \cap R|$  elements from
$T_{\alpha' - \beta+1} \dots T_{\alpha'}$
and $|S \cap \bar{R}| \leq  \alpha' - \beta$  elements from $T_{1}\dots T_{\alpha' - \beta}$.
Finally, $T(S) \leq \beta + |S \cap \bar{R}|$ since at most
all elements in $S \cap \bar{R}$ can be taken from $T_{1}\dots T_{\alpha' - \beta}$ and at most
$\beta$ elements in $S \cap R$  from $T_{\alpha' - \beta+1} \dots T_{\alpha'}$.
\end{itemize}
\end{proof}
\end{document}